%% file: parmerlin.tex
\documentclass{superfri}
\usepackage{amssymb}
\usepackage{fancyvrb}
\usepackage{dsfont}
\usepackage{amsmath}
\usepackage{amsthm}
\usepackage{amsfonts}
\usepackage{marvosym}
\usepackage{multirow}

\usepackage{setspace}
\usepackage{graphicx}
\usepackage[hidelinks]{hyperref}

\DeclareGraphicsExtensions{.pdf,.png,.jpg}
\graphicspath{{pics/},{plots/},
	{plots/casestudies/ECG/discords/},{plots/casestudies/ECG/heatmap/},
	{plots/casestudies/ECG2/discords/},{plots/casestudies/ECG2/heatmap/},
	{plots/casestudies/PolyTER/aud346/discords/},{plots/casestudies/PolyTER/aud346/heatmap/},
	{plots/casestudies/PolyTER/aud434/discords/},{plots/casestudies/PolyTER/aud434/heatmap/},
	{plots/casestudies/PolyTER/aud808/discords/},{plots/casestudies/PolyTER/aud808/heatmap/},
	{plots/casestudies/PowerDemand/discords/},{plots/casestudies/PowerDemand/heatmap/},
	{plots/casestudies/Respiration/discords/},{plots/casestudies/Respiration/heatmap/},
	{plots/casestudies/SpaceShuttle/discords/}{plots/casestudies/SpaceShuttle/heatmap/}
}

\usepackage{caption}
\usepackage{subcaption}

\usepackage{listings}
\usepackage{algorithm}
\usepackage[noend]{algpseudocode}

\usepackage[dvipsnames]{xcolor}
\usepackage{pgfplots} 
\usetikzlibrary{patterns}

\definecolor{green1}{HTML}{9FD4A3}
\definecolor{green2}{HTML}{AEDCAE}
\definecolor{green3}{HTML}{BEE3BA}
\definecolor{green4}{HTML}{CDEBC5}
\definecolor{green5}{HTML}{DDF2D1}
\definecolor{green6}{HTML}{ECFADC}

\usepackage{dsfont}

\DeclareUnicodeCharacter{2061}{}
\DeclareUnicodeCharacter{430}{}

\def\ParMERLIN{PALMAD}
\def\ParMERLINtranscript{Parallel Arbitrary Length MERLIN-based Anomaly Discovery}
\def\ParDRAG{PD3}
\def\ParDRAGtranscript{Parallel DRAG-based Discord Discovery}

\begin{document}
	
\author{Mikhail~L. Zymbler\footnote{\label{susu}South Ural State University, Chelyabinsk, Russian Federation}
	\and  Yana~A. Kraeva\footnoteref{susu}
}

\title{High-performance Time Series Anomaly Discovery \\on Graphics Processor}

\maketitle{}

\begin{abstract}

Currently, discovering subsequence anomalies in time series remains one of the most topical research problems. A subsequence anomaly refers to successive points in time that are collectively abnormal, although each point is not necessarily an outlier. Among a large number of approaches to discovering subsequence anomalies, the discord concept is considered one of the best. A time series discord is intuitively defined as a subsequence of a given length that is maximally far away from its non-overlapping nearest neighbor. Recently introduced the MERLIN algorithm discovers time series discords of every possible length in a specified range, thereby eliminating the need to set even that sole parameter to discover discords in a time series. However, MERLIN is serial and its parallelization could increase the performance of discords discovery. In this article, we introduce a novel parallelization scheme for GPUs, called \ParMERLIN{}, \ParMERLINtranscript{}. As opposed to its serial predecessor, \ParMERLIN{} employs recurrent formulas we have derived to avoid redundant calculations, and advanced data structures for the efficient implementation of parallel processing. Experimental evaluation over real-world and synthetic time series shows that our algorithm outperforms parallel analogs. We also apply \ParMERLIN{} to discover anomalies in a real-world time series employing our proposed discord heatmap technique to illustrate the results.

\end{abstract}

\input{parmerlin-introduction.tex}
\input{parmerlin-relatedwork.tex}
\input{parmerlin-preliminaries.tex}

\input{parmerlin-approach.tex}

\input{parmerlin-experiments.tex}

\input{parmerlin-casestudies.tex}

\input{parmerlin-conclusion.tex}

\section*{Acknowledgment}
\label{sec:Ack}

This work was financially supported by the Russian Science Foundation (grant No.~23-21-00465). The research is carried out using the equipment of the shared research facilities of HPC computing resources at Lomonosov Moscow State University and supercomputer resources of the South Ural State University.


\bibliography{parmerlin}

\linespread{0.6}\selectfont
\small
\input{parmerlin-appendix.tex}

\end{document}

%% file: parmerlin-introduction.tex

\section*{Introduction}
\label{sec:Introduction}

Over past decades, time series data are ubiquitous in diverse spheres of a human's activity: industry, healthcare, science, social, and so on. Currently, discovering anomalies (or outliers as a synonym) in time series remains one of the most topical research problems. In a 
time series, 
point and subsequence anomalies can serve as the aims to be detected~\cite{DBLP:journals/csur/Blazquez-Garcia21}. The former defines a datum that deviates in a specific time instant when compared either to the other values in the time series or to its neighboring points. The latter refers to successive points in time whose collective behavior is unusual, although each observation individually is not necessarily a point anomaly. Subsequence anomaly detection is more challenging due to the need to take into account the subsequence length among other aspects~\cite{DBLP:journals/csur/Blazquez-Garcia21}. 

Among a wide spectrum of analytical and neural network-based approaches to time series subsequence anomaly detection~\cite{DBLP:journals/csur/Blazquez-Garcia21,DBLP:journals/access/ChoiYPY21}, the discord concept~\cite{DBLP:conf/cbms/LinKFH05} is considered one of the best~\cite{DBLP:journals/csur/ChandolaBK09,ChandolaCK09}. A time series discord is intuitively defined as a subsequence that is maximally far away from its non-overlapping nearest neighbor. Discords look attractive for an end-user since they require the only parameter to be specified, the subsequence length. However, application of discords is reduced by sensitivity to this single user choice. A straightforward solution of this problem, namely discovering discords of all the possible lengths and then selecting the best discords with respect to some measure looks computationally prohibitive. 

Nevertheless, recently introduced the MERLIN algorithm~\cite{DBLP:conf/icdm/NakamuraIMK20} can efficiently and exactly discover discords of every possible length in a specified range, being ahead of competitors in terms of accuracy and performance. Thus, MERLIN allows an end-user for removing the need to set even that above-mentioned sole parameter to discover discords in a time series. MERLIN employs repeated calls of the DRAG algorithm~\cite{DBLP:conf/icdm/YankovKR07} that discovers discords of a given length with a distance of at least $r$ to their nearest neighbors, and adaptive selection of the parameter~$r$. However, multiple calls of the above sub-algorithm result in calculations that are partially repeated and being executed at once would increase the performance of MERLIN. Furthermore, the MERLIN algorithm is serial and looks attractive for  parallelization to increase the performance of discords discovery.

In this study, we address the problem of parallelization the MERLIN algorithm for discovery of arbitrary length discords on GPU continuing our research on accelerating various time series mining tasks with parallel architectures and in-database time series analysis~\cite{DBLP:conf/rcdl/KraevaZ18a,Zymbler19,ZymblerK19,ZymblerGKK21,ZymblerPK19,ZymblerK20,ZymblerI21,ZymblerG22}. The article's contribution can be summarized as follows. We perform thorough review of works related to the discord-based approaches to discovering time series anomalies and their parallelization for diverse hardware platforms. Next, based on MERLIN~\cite{DBLP:conf/icdm/NakamuraIMK20}, we introduce the parallel algorithm \ParMERLIN{} (\ParMERLINtranscript{}) for discovery of arbitrary length discords on a graphics processor. \ParMERLIN{} employs recurrent formulas we have derived to avoid redundant calculations, and advanced data structures for the efficient implementation of parallel processing. Further, in the extensive experimental evaluation over real-world and synthetic time series, we show that our algorithm outperforms existent analogs. Finally, we apply \ParMERLIN{} to anomaly discovery in a real-world time series employing our proposed discord heatmap technique to illustrate the results. 

The remainder of the article is organized as follows. In Section~\ref{sec:RelatedWork}, we discuss related works. Section~\ref{sec:Preliminaries} contains notation and formal definitions along with a short description of the original serial algorithm. Section~\ref{sec:Approach} introduces the proposed parallel algorithm to discover time series discords on GPU. In Section~\ref{sec:Experiments}, we give the results and discussion of the experimental evaluation of our algorithm.  Section~\ref{sec:CaseStudies} describes a case study on discord discovery in a real-world time series. Finally, in \nameref{sec:Conclusion}, we summarize the results obtained and suggest directions for further research.

%% file: parmerlin-relatedwork.tex

\section{Related Work}
\label{sec:RelatedWork}

The time series discord concept was introduced by Keogh \textit{et~al.}~\cite{DBLP:conf/cbms/LinKFH05}, and currently, is considered as one of the best analytical approaches to discovering anomalies in time series~\cite{DBLP:journals/csur/ChandolaBK09,ChandolaCK09}. Time series discord is intuitively defined as the subsequence of a time series that is the most distant to its non-overlapping nearest neighbor. Discords look attractive for an end-user since they require the only parameter to be specified, the subsequence length. Below, we consider discord-based research on discovering anomalies including studies addressed the parallelization of the approaches above.

In~\cite{DBLP:conf/cbms/LinKFH05}, Keogh \textit{et~al.} proposed the HOTSAX (Heuristically Ordered Time series using Symbolic Aggregate
Ap\-pro\-Xi\-ma\-tion) algorithm for discords discovery in a time series that can be entirely placed in RAM. HOTSAX employs time series encoding through the SAX technique~\cite{DBLP:conf/dmkd/LinKLC03} and the Euclidean distance. HOTSAX iterates through all the pairs of subsequences calculating the distance between them, and finds the maximum among the distances to the nearest neighbor. The algorithm employs a prefix trie~\cite{DBLP:journals/cacm/Fredkin60} to index the subsequences. When iterating, unpromising subsequences are discarded without calculating distances. The subsequence with a neighbor closer than the best-so-far maximum of distances to all the nearest neighbors is unpromising. HOTSAX exploits a certain heuristic that allows to discarding more unpromising candidates. Improvements of HOTSAX include \textit{i}SAX~\cite{DBLP:conf/kdd/ShiehK08} and HOT-\textit{i}SAX~\cite{DBLP:conf/kse/BuuA11} (indexable SAX), WAT~\cite{DBLP:conf/sdm/BuLFKPM07,DBLP:conf/adma/FuLKL06} (application of the Haar wavelets instead of SAX and augmented trie), HashDD~\cite{ThuyAC16} (employing a hash table instead of the prefix trie), HDD-MBR~\cite{ChauDA18} (application of R-trees), 
BitClusterDiscord~\cite{DBLP:journals/kbs/LiBJWW13} (employing clustering of the bit representation of subsequences), and HST (HOTSAX Time)~\cite{AvogadroD22} (reduction the size of the discord search space through the warm-up process and the similarity between subsequences close in time). 

In~\cite{DBLP:conf/edbt/Senin0WOGBCF15}, Senin \textit{et~al.} proposed the HOTSAX-based RRA (Rare Rule Anomaly) algorithm to discover variable-length discords. RRA deals with a time series discretized with SAX applying grammar-induction procedures. Since symbols infrequently used in grammar rules are non-repetitive and thus potentially unusual, the discords correspond to infrequent grammar rules that vary in length. Taking into account that the lengths of the subsequences vary, the distance between them is calculated by shrinking the longest subsequence with the Piecewise Aggregate Approximation (PAA)~\cite{DBLP:journals/kais/KeoghCPM01} to obtain subsequences of the same length. Although the RRA algorithm is a step forward from HOTSAX to  parameter-free discord discovering, as its predecessor, it is limited to RAM-stored time series. 

In~\cite{DBLP:conf/icdm/YankovKR07}, Yankov, Keogh \textit{et~al.} presented the DRAG (Discord Range Aware Gathering) algorithm for discovering discords in a time series stored on a disk rather than in RAM. DRAG introduces the range discord concept, where such a discord has a distance of at least $r$ to its non-overlapping nearest neighbor, and $r$ is a user-defined threshold. The DRAG algorithm performs in two phases, namely the candidate selection (collecting potential range discords) and discord refinement (discarding false positives), with each phase requiring one linear scan through the time series on the disk. In~\cite{Son20}, Son slightly improved DRAG through the employing a hash bucket data structure to speed up the candidate selection phase. The authors of DRAG 
proposed 
the following procedure to choose the parameter $r$. Through the uniform sampling, one can obtain a maximum length fragment of the original time series that fits in RAM. Next, the HOTSAX algorithm discovers a discord in the above-obtained fragment. Finally, the $r$ threshold is assumed to be equal to the distance of the discord found to its nearest neighbor. 

However, in the DRAG algorithm, the above-described heuristic does not define a formal way to choose the parameter $r$ to guarantee the efficiency of discords discovery~\cite{DBLP:conf/sigmod/MueenNL10}. Ideally, $r$ should be set in such a way that it is a little less than the distance between the  discord eventually found and its nearest neighbor~\cite{DBLP:conf/icdm/NakamuraIMK20}. Then the time and space complexity of DRAG is $O(mn)$, where $n$ is the length of the time series, and $m$ is the discord length. If the value of $r$ is set significantly less than the above-mentioned distance, then the algorithm will find the discord, but the time and spatial complexity will be higher, $O(n^2)$. Finally, if $r$ is greater than the above distance, then no discords will be found. In addition, DRAG (like HOTSAX as its predecessor) is not able to discover \emph{all} the discords in the following meaning: the algorithm finds discords of a single specified length, but not discords of every possible length in a specified range. In the latter case, a brute-force approach involving a cyclic runs of the DRAG algorithm for a specified range of the discord length does not work, since at each iteration of such a loop, we should choose the parameter $r$ from scratch.

Recently introduced by Keogh \textit{et~al.} the MERLIN algorithm~\cite{DBLP:conf/icdm/NakamuraIMK20} overcomes the above-described limitations of DRAG. MERLIN calls DRAG repeatedly and adaptively selects the parameter $r$. In the experiments, MERLIN efficiently and exactly discovers discords of every possible length, being ahead of competitors both in accuracy and performance~\cite{DBLP:conf/icdm/NakamuraIMK20}. The authors also empirically showed that MERLIN is able to discover point, contextual, and collective anomalies according to the taxonomy in~\cite{DBLP:journals/csur/ChandolaBK09}. Moreover, the authors mentioned that despite the recent explosion of the deep learning anomaly detection methods (e.g.,~\cite{DBLP:journals/ijon/AhmadLPA17,DBLP:conf/kdd/HundmanCLCS18,DBLP:conf/icmla/FarahaniCKKK19,DBLP:journals/access/MunirSDA19}), it is not obvious that they outperform discord-based approaches, since the former, by its nature, require many critical parameters to be set whereas the latter are domain-independent and require one intuitive parameter the need to set which can even be removed by MERLIN. However, the MERLIN algorithm is still serial, and its parallelization (for various hardware platforms) could increase the performance of discords discovery. In addition, let us mention that, in MERLIN, multiple calls of DRAG result in calculations that are partially repeated (e.g., normalization of the subsequences in a specified range of length) and being executed at once would also increase the performance of MERLIN.

Research addressed the problem of parallelization the discord discovery include the following. In~\cite{ZymblerPK19,Zymbler19}, Zymbler \textit{et~al.} proposed a parallelization schema for HOTSAX to discover discords with Intel many-core processors or GPUs through the OpenMP~\cite{DBLP:journals/pieee/SupinskiSDKBOTM18} or OpenACC~\cite{DBLP:journals/tjs/ReyesLFS13} technology, respectively. The algorithm employs matrix data layout to organize calculations with as many vectorizable loops as possible. Similarly to its predecessor, the algorithm differs the following sets of subsequences: ones with the least frequent SAX words and the rest, and for any subsequence~-- ones whose SAX words match the given subsequence's SAX word and the rest. When iterating all the subsequences through two nested loops, the algorithm parallelizes separately and differently for the outer and the inner loops, depending on the number of running threads and the cardinality of the above-mentioned sets. 

In~\cite{DBLP:journals/kais/YankovKR08} (an expanded version of~\cite{DBLP:conf/icdm/YankovKR07}), Yankov, Keogh \textit{et~al.} discussed the parallel version of DRAG based on the MapReduce paradigm~\cite{DBLP:conf/osdi/DeanG04}, and the key idea is as follows. Let the input time series be partitioned evenly across $P$ high-performance cluster nodes. Each node selects candidates in its own partition with the same parameter $r$ resulting in the local candidate set $\mathcal{C}_i$. Then the global candidate set $\mathcal{C}$ is constructed as $\mathcal{C}=\cup^P_{i=1}\mathcal{C}_i$ and sent to each cluster node. Next, a node refines candidates in its own partition taking the global candidate set $\mathcal{C}$ as an input, and produces the local refined candidate set $\tilde{\mathcal{C}}_i$. Finally, the global discords set $\mathcal{D}$ is computed as $\mathcal{D}=\cap^P_{i=1}\tilde{\mathcal{C}}_i$. In the experimental evaluation, the authors, however, just simulated the above-mentioned scheme on up to eight computers resulting in a close-to-linear speedup.

In~\cite{ZymblerGKK21}, Zymbler \textit{et~al.} introduced a parallelization scheme of DRAG for a high-performance cluster with Intel many-core processors. For the cluster node, the authors define matrix data structures and employ thread-level parallelism through the OpenMP technology~\cite{DBLP:journals/pieee/SupinskiSDKBOTM18} whereas communication among the cluster nodes is implemented through MPI (Message Passing Interface)~\cite{DBLP:journals/cacm/Snir18}. As 
opposed to DRAG parallelization scheme, at each cluster node, the authors first refine the local candidate set $\mathcal{C}_i$ with respect to the same parameter $r$ resulting in the $\tilde{\mathcal{C}}_i$ set, and then construct the global candidate set as $\mathcal{C}=\cup^P_{i=1}\tilde{\mathcal{C}}_i$ relying on the fact that a candidate is not a true discord if it was pruned by at least one cluster node during the selection phase. In the experiments~\cite{ZymblerGKK21}, the authors showed that such a technique allows for significant reduction of the global candidate set and increasing the overall algorithm's performance. The algorithm also significantly outperforms the following DRAG-based parallel discords discovery algorithms for high-performance clusters: DDD (Distributed Discord Discovery)~\cite{DBLP:conf/hpcc/WuZHLLL15} and PDD (Parallel Discord Discovery)~\cite{DBLP:conf/pakdd/HuangZMLLWHD16}. The above-mentioned competitors are far behind since the fact that they involve intensive data exchanges across cluster nodes. However, the above-described parallelization still cannot efficiently discover discords of every possible length in a specified range.

In the review, we should also mention the matrix profile (hereinafter MP) concept proposed by Keogh \textit{et~al.}~\cite{DBLP:journals/datamine/YehZUBDDZSMK18}. For a given time series, MP can informally be defined as a time series, where the $i$\nobreakdash-th element is the distance from the $i$\nobreakdash-th subsequence of the original time series to its non-overlapping nearest neighbor. MP plays the role of a building block the solutions of various time series motif discovery related problems are based on (semantic motifs~\cite{DBLP:conf/icdm/ImaniK19}, snippets~\cite{DBLP:journals/datamine/ImaniMDCK20}, chains~\cite{DBLP:journals/kais/ZhuINK19}, etc.). According to the above definition, top\nobreakdash-$k$ discords can be discovered as a by-product of the MP calculation since they are the subsequences on which the top-$k$ maximum values in MP are achieved. However, the time complexity of MP computation is high, namely $O(n^2)$ (where $n$ is the time series length)~\cite{DBLP:journals/datamine/YehZUBDDZSMK18,DBLP:conf/icdm/0014YZKK18}, so straightforward employing of MP in the discords discovery results in low performance as it has been evaluated in the following experiments. The serial SCRIMP algorithm~\cite{DBLP:journals/datamine/YehZUBDDZSMK18} is inferior to MERLIN~\cite{DBLP:conf/icdm/NakamuraIMK20}. Parallel MP algorithms for graphics processor and high-performance cluster, GPU\nobreakdash-STAMP~\cite{DBLP:journals/datamine/YehZUBDDZSMK18} and MP\nobreakdash-HPC~\cite{Pfeilschifter19}, respectively, are inferior to the parallel discords discovery algorithm for high-performance cluster with Intel many-core processors proposed in~\cite{ZymblerGKK21}.

In~\cite{ThuyAC21}, Thuy \textit{et~al.} introduced the notion of the  $K$\nobreakdash-distance discord, namely a subsequence with the largest sum of distances to its non-overlapping $K$ nearest neighbors. Such an approach aims at solving so-called ``twin freak'' problem~\cite{WeiKX06} when a discord fails to discover an anomalous (rare) subsequence if it occurs more than once in the time series, and is a modification of the J\nobreakdash-distance discord~\cite{HuangZWS15} concept, where the distance between a subsequence and its $k$\nobreakdash-th non-overlapping nearest neighbor is employed. The authors also presented the KBF\_GPU (Brute-Force for K\nobreakdash-distance discord)  algorithm that accelerates $K$\nobreakdash-distance discord discovery on a graphics processor. KBF\_GPU iterates all the subsequences of a given time series through two nested loops where the inner loop is parallelized and adapted to calculate sum of distances. In the experiments, the authors, however, compare their algorithm only with serial HOTSAX~\cite{DBLP:conf/cbms/LinKFH05}, and the latter, as expected, is significantly inferior to KBF\_GPU. 

In~\cite{DBLP:journals/tpds/ZhuJGD21}, Zhu \textit{et~al.} presented a parallel algorithm to accelerate discords discovery with GPU. The authors exploit the normalized Euclidean distance and its efficient calculation through the Pearson correlation by the technique proposed in~\cite{DBLP:conf/sigmod/MueenNL10}. To provide high performance of discord discovery, the algorithm employs two computational patterns. The first one prescribes the following two-step procedure. First, calculate the minimum distance between the discord candidate subsequence and all other subsequences of the time series that do not overlap the candidate. Then, find a candidate on which the maximum distance among all the candidates is achieved. The second pattern assumes an early stop of calculations in the pattern above when the distance between the candidate and a certain subsequence is less than the best-so-far distance. In such a case, both the candidate and the subsequence are obviously not discords, and we do not need to calculate distances from the candidate to other non-overlapping subsequences. In the experiments~\cite{DBLP:journals/tpds/ZhuJGD21}, the proposed algorithm outruns SCAMP~\cite{DBLP:conf/cloud/ZimmermanKSCFBK19} that is currently the fastest parallel algorithm for calculating the matrix profile. However, the proposed computational patterns limit the result to a single (albeit the most important) discord of the time series, whereas the above-described algorithms are based on the range discord concept and able to discover top-$k$ discords, where the parameter $k$ is prespecified by an expert in the subject domain.

Concluding our overview of related work, it can be seen that, currently, the MERLIN algorithm~\cite{DBLP:conf/icdm/NakamuraIMK20} based on the range discord concept~\cite{DBLP:conf/icdm/YankovKR07} is one of the most promising  approaches to discover anomalies in time series. Moreover, being analytical and agnostic, MERLIN is at least competitive with deep learning methods. However, parallelization of MERLIN could increase the performance of discord discovery. Such parallelization is a topical issue since, to the best of our knowledge, no research has addressed the accelerating discovery of discords of every possible length with GPU or any other parallel hardware architecture.

%% file: parmerlin-preliminaries.tex

\section{Preliminaries}
\label{sec:Preliminaries}

Prior to detailing the proposed parallel algorithm for discords discovery, in Sections~\ref{subsec:NotationDefinitions} and~\ref{subsec:MERLINDRAG}, we introduce basic notation and formal definitions according to~\cite{DBLP:conf/icdm/NakamuraIMK20,DBLP:conf/icdm/YankovKR07} and give an overview of the original serial algorithms MERLIN and DRAG our development is based on, respectively.

\input{parmerlin-definitions.tex}
\input{parmerlin-merlindrag.tex}

%% file: parmerlin-definitions.tex

\subsection{Notation and Definitions}
\label{subsec:NotationDefinitions}

A \emph{time series} is a chronologically ordered sequence of real-valued numbers: 
\begin{equation}
	\label{eq:TimeSeries}
	T=\{t_i\}_{i=1}^{n}, \quad t_i \in \mathds{R}.
\end{equation}
The length of a time series, $n$, is denoted by $\lvert T\rvert$. Hereinafter, we assume that the time series $T$ fit into the main memory.

A \emph{subsequence} $T_{i,\,m}$ of a time series $T$ is its subset of $m$ successive elements that starts at the $i$-th position:
\begin{equation}
	\label{eq:Subsequence}
	T_{i,\,m}=\{t_k\}_{k=i}^{i+m-1}, \quad 1 \leq i \leq n-m+1, \quad 3 \leq m \ll n.  
\end{equation}
We denote the set of all $m$-length subsequences in $T$ by $S^m_T$. Let $N$ denotes the number of subsequences in $S^m_T$, i.e., $N=|S^m_T|=n-m+1$.

A \emph{distance function} for any two $m$-length subsequences is a nonnegative and symmetric function $\operatorname{Dist}: \mathds{R}^m \times \mathds{R}^m \rightarrow \mathds{R}$.

Given a time series $T$ and its two subsequences $T_{i,\,m}$ and $T_{j,\,m}$, we say that they are \emph{non-self match} to each other at distance $\operatorname{Dist}(T_{i,\,m},\;T_{j,\,m})$ if $|i-j| \geq m$. Let us denote a non-self match of a subsequence $C \in S^m_T$ by $M_C$.

Given a time series $T$, its subsequence $D \in S^m_T$ is said to be the \emph{discord} if $D$ has the largest distance to its nearest non-self match. Formally speaking, the discord $D$ meets the following:
\begin{equation}
	\label{eq:Discord} 
	\forall C \in S^m_T \quad \min\bigl(\operatorname{Dist}(D,\,M_D)\bigr) > \min \bigl(\operatorname{Dist}(C,\,M_C)\bigr).
\end{equation}

The definition above is generalized from top-1 to top-$k$ discord as follows: $D \in S^m_T$ is said to be the $k$-th discord if the distance to its $k$-th nearest non-self match is the largest.

Given the positive real number $r$, the discord at a distance at least $r$ from its nearest non-self match is called the \emph{range discord}. That is, the range discord $D$ with respect to the parameter $r$ meets the following: $\min\bigl(\operatorname{Dist}(D,\,M_D)\bigr) \geq r$.

The MERLIN~\cite{DBLP:conf/icdm/NakamuraIMK20} and DRAG~\cite{DBLP:conf/icdm/YankovKR07} algorithms deal with subsequences of the time series that previously z\nobreakdash-normalized to have mean zero and a standard deviation of one. Here, z\nobreakdash-normalization of a subsequence $X=\{x\}_{i=1}^{m} \in S^m_T$  is defined as a subsequence $\hat{X}=\{\hat{x}\}_{i=1}^{m}$, where
\begin{equation}
	\label{eq:zNormalization}
	\begin{gathered}
		\hat{x}_i = \frac{x_{i}-\mu_{X}}{\sigma_{X}}, \quad 	
		\mu_{X}=\frac{1}{m}\sum\limits_{i=1}^{m} x_{i}, \quad 
		\sigma^2_{X}=
		{\frac{1}{m}\sum\limits_{i=1}^{m}x_i^2-\mu^2}.
	\end{gathered}
\end{equation}	

Both MERLIN and DRAG algorithms employ the Euclidean metric as the $\operatorname{Dist}(\cdotp,\cdotp)$ function to measure the distance between  subsequences that is defined as follows. Let us have $X,Y \in S^m_T$, then the Euclidean distance between the subsequences is calculated as below:  
\begin{equation}
	\label{eq:Euclid}
	\operatorname{ED}(X,\;Y)=\sqrt{\sum\limits_{i=1}^m(x_{i}-y_{i})^2}.
\end{equation}

In our study, being motivated by the highest possible performance of discord discovery, we employ the square of the Euclidean metric as a distance function. For the sake of simplicity, we denote by $\operatorname{ED_{norm}}$ the Euclidean distance between two z\nobreakdash-normalized subsequences: $\operatorname{ED_{norm}}(X,\,Y)=\operatorname{ED}(\hat{X}, \, \hat{Y})$. To compute $\operatorname{ED^2_{norm}}$, we further employ the following technique proposed in~\cite{DBLP:conf/sigmod/MueenNL10} that allows for faster calculation than in Equation~\ref{eq:Euclid}:
\begin{equation}
	\label{eq:EDnorm}
	\operatorname{ED^2_{norm}}(X,\,Y) = 2m\left(1-\dfrac{X \cdotp Y - m \cdotp \mu_X \cdotp \mu_Y}{m \cdotp \sigma_X \cdotp \sigma_Y}\right),
\end{equation}
where $X \cdotp Y$ denotes the scalar product of vectors $X,Y \in \mathds{R}^m$.

%% file: parmerlin-merlindrag.tex

\subsection{MERLIN and DRAG Algorithms}
\label{subsec:MERLINDRAG}

\begin{algorithm}[!ht]
	\caption{\textsc{MERLIN}
		(\textsc{in} $T$, $\textit{minL}$, $\textit{maxL}$, $\textit{topK}$; \textsc{out} $\mathcal{D}$)
	}
	\begin{algorithmic}[1]
		\State{$\mathcal{D} \leftarrow \varnothing$; $r \leftarrow 2\sqrt{\textit{minL}}$; $\textit{nnDist}_{\textit{minL}} \leftarrow -\infty$}
		\While{$\textit{nnDist}_{\textit{minL}}<0$  \textbf{and} $|D_{\textit{minL}}| < \textit{topK}$}
			\State{$D_{minL} \leftarrow \hyperref[alg:DRAG]{\textsc{DRAG}}(T, minL, r)$; $\mathcal{D} \leftarrow \mathcal{D} \cup D_{minL}$; $\textit{nnDist}_{\textit{minL}} \leftarrow \min\limits_{d \in D_{\textit{minL}}} d.\textit{nnDist}$}
			\State{$r \leftarrow 0.5 \cdot r$}
		\EndWhile
		\For{$i \leftarrow \textit{minL}+1$ \textbf{to} $\textit{minL}+4$}
			\State{$\textit{nnDist}_{i} \leftarrow -\infty$}
			\While{$\textit{nnDist}_{i}<0$  \textbf{and} $|D_i| < \textit{topK}$}
				\State{$r \leftarrow 0.99 \cdot nnDist_{i-1}$}
				\State{$D_{i} \leftarrow \hyperref[alg:DRAG]{\textsc{DRAG}}(T, i, r)$; $\mathcal{D} \leftarrow \mathcal{D} \cup D_{i}$; $\textit{nnDist}_{i} \leftarrow \min\limits_{d \in D_{i}} d.\textit{nnDist}$}				
				\State{$r \leftarrow 0.99 \cdot r$}
			\EndWhile
		\EndFor
		\For{$i \leftarrow \textit{minL}+5$ \textbf{to} $\textit{maxL}$}
			\State{$\mu \leftarrow \operatorname{Mean}(\{\textit{nnDist}_k\}^{i-5}_{k=i-1})$; $\sigma \leftarrow \operatorname{Std}(\{\textit{nnDist}_k\}^{i-5}_{k=i-1})$; $r \leftarrow \mu-2\sigma$}
			\State{$D_{i} \leftarrow \hyperref[alg:DRAG]{\textsc{DRAG}}(T, i, r)$; $\mathcal{D} \leftarrow \mathcal{D} \cup D_{i}$; $\textit{nnDist}_{i} \leftarrow \min\limits_{d \in D_{i}} d.\textit{nnDist}$}
			\While{$nnDist_{i}<0$  \textbf{and} $|D_i| < \textit{topK}$}
				\State{$D_{i} \leftarrow \hyperref[alg:DRAG]{\textsc{DRAG}}(T, i, r)$; $\mathcal{D} \leftarrow \mathcal{D} \cup D_{i}$; $\textit{nnDist}_{i} \leftarrow \min\limits_{d \in D_{i}} d.\textit{nnDist}$}
				\State{$r \leftarrow r - \sigma$}
			\EndWhile
		\EndFor
		\State\Return{$\mathcal{D}$}
	\end{algorithmic}		
	\label{alg:MERLIN}
\end{algorithm}

Alg.~\ref{alg:MERLIN} depicts a pseudo code of MERLIN~\cite{DBLP:conf/icdm/NakamuraIMK20} (up to discovering top-$k$ discords of each length instead of all ones in the specified length range). Hereinafter, let us have an $n$\nobreakdash-length time series $T$, and we are to find a set $\mathcal{D}$ of its discords that have a length in the range $\textit{minL}..\textit{maxL}$ (where $\textit{minL} \le \textit{maxL} \ll n$), so that $\mathcal{D}=\cup_{m=\textit{minL}}^{\textit{maxL}}D_m$, where $D_m$ denotes a subset of $m$\nobreakdash-length discords. The distance to the $d \in \mathcal{D}$ discord's nearest neighbor is denoted by $d.\textit{nnDist}$.

The algorithm prescribes the following procedure to select the parameter $r$. Discords are discovered sequentially, starting from a minimum length of the specified discord range to a maximum one. At each step, MERLIN calculates the arithmetic mean $\mu$ and the standard deviation $\sigma$ of the last five distances from the discords found to their nearest neighbors, and then calls the DRAG algorithm passing it the parameter $r=\mu-2\sigma$. If DRAG has not found a discord, then $\sigma$ is subtracted from $r$ until DRAG stops successfully (i.e., a discord will be found). For the first five discord lengths, the parameter $r$ is set as follows. For discords of minimum length $\textit{minL}$, the parameter is set as $r=2\sqrt{\textit{minL}}$ since it is the maximum possible distance between any pair of $\textit{minL}$-length subsequences, and then $r$ is reduced by half until DRAG with such a parameter results in success. To obtain the next four discord lengths, the algorithm takes the distance from the discord to its nearest neighbor obtained in the previous step, minus a small value equal to 1\%. The subtraction of an additional 1\% proceeds until the discord discovery with such a parameter results in success. For a detailed explanation of the above-described procedure, we refer the reader to the original work~\cite{DBLP:conf/icdm/NakamuraIMK20}.

\begin{algorithm}[!ht]
	\caption{\textsc{DRAG}
		(\textsc{in} $T$, $m$, $r$; \textsc{out} $\mathcal{D}$)
	}	\begin{minipage}[t]{0.45\columnwidth}%
		\centering\textsc{Phase~1. Select candidates }
		\begin{algorithmic}[1]
			\State{$\mathcal{C} \leftarrow \{T_{1,\,m}\}$}
			\ForAll{$s \in S^m_T \smallsetminus T_{1,\,m}$}
			\State{$\textit{isCand} \leftarrow$ \texttt{TRUE}}
			\ForAll{$c \in \mathcal{C} \; \textbf{and} \; c \in M_s$}
			\If{$\hyperref[eq:Euclid]{\operatorname{ED}}(s,c)<r$}
			\State{$\mathcal{C} \leftarrow \mathcal{C} \smallsetminus c$}
			\State{$\textit{isCand} \leftarrow$ \texttt{FALSE}}
			\EndIf
			\EndFor
			\If{$ \textit{isCand} $}
			\State{$\mathcal{C} \leftarrow \mathcal{C} \cup s$}
			\EndIf
			\EndFor
			\State\Return{$\mathcal{C}$}
		\end{algorithmic}
	\end{minipage}
	\hfill\vrule\hfill
	\begin{minipage}[t]{0.54\columnwidth}%
		\centering\textsc{Phase~2. Refine discords }
		\begin{algorithmic}[1]
			\State{$\mathcal{D} \leftarrow \varnothing$; $\forall c \in \mathcal{C} \; c.\textit{nnDist} \leftarrow +\infty$}
			\ForAll{$s \in S^m_T$}
				\ForAll{$c \in \mathcal{C} \; \textbf{and} \; c \in M_s \; \textbf{where} \; s \neq c$}
					\State{$dist \leftarrow \operatorname{EarlyAbandon\hyperref[eq:Euclid]{\operatorname{ED}}}(s,c)$}
					\If{$dist < r$}
						\State{$\mathcal{C} \leftarrow \mathcal{C} \smallsetminus c$}
					\Else					
						\State{$\mathcal{D} \leftarrow \mathcal{D} \cup c$}
						\State{$c.\textit{nnDist} \leftarrow \operatorname{min}(c.\textit{nnDist}, \textit{dist})$}
					\EndIf
				\EndFor
			\EndFor
			\State\Return{$\mathcal{D}$}
		\end{algorithmic}
	\end{minipage}
	\label{alg:DRAG}
\end{algorithm}

The DRAG algorithm~\cite{DBLP:conf/icdm/YankovKR07} (see Alg.~\ref{alg:DRAG}) performs in two phases, namely the candidate selection and discord refinement, where it collects potential range discords and discards false positives, respectively. At the first phase, DRAG scans through the time series $T$, and for each subsequence $s \in S^m_T$ it validates the possibility for each candidate $c$ already in the candidate set $\mathcal{C}$ to be discord. If a candidate $c$ fails the validation, then it is removed from this set. In the end, the new $s$ is either added to the candidates set, if it is likely to be a discord, or it is pruned. At the second phase, the algorithm initially sets distances of all candidates to their nearest neighbors to positive infinity. Then, DRAG scans through the time series $T$, calculating the distance between each subsequence $s \in S^m_T$ and each candidate $c$. When calculating $\operatorname{ED}(s,c)$, the $\operatorname{EarlyAbandonED}$ procedure stops the summation of $\sum^m_{k=1}(s_k-c_k)^2$ if it reaches $k=\ell$, such that $1 \leq \ell \leq m$ for which $\sum^\ell_{k=1}(s_k-c_k)^2 \ge c.nnDist^2$. If the distance is less than $r$ then the candidate is false positive and permanently removed from $\mathcal{C}$. If the above-mentioned distance is less than the current value of $c.nnDist$ (and still greater than $r$, otherwise it would have been removed) then the current distance to the nearest neighbor is updated. The correctness of the above-described procedure is proved in the original work~\cite{DBLP:conf/icdm/YankovKR07}.

%% file: parmerlin-approach.tex

\section{Arbitrary Length Discords Discovery with GPU}
\label{sec:Approach}

Currently, GPU (Graphics Processing Unit)~\cite{DBLP:conf/iwmm/Kirk07} is one of the most popular many-core hardware platforms. GPU fits well for SIMD (Single Instructions Multiple Data) computations being composed of symmetric streaming multiprocessors, each of wich, in turn, consists of symmetric CUDA (Compute Unified Device Architecture) cores. CUDA API (Application Programming Interface) makes it possible to assign multiple threads to execute the same set of instructions over multiple data. In CUDA, all threads form a \emph{grid} consisting of blocks. In a \emph{block}, threads are divided into \emph{warps}, logical groups of 32~threads. The block's threads run in parallel and communicate with each other through shared memory. A CUDA function is called a \emph{kernel}. When run a kernel on GPU, an application programmer specifies both the number of blocks in the grid and the number of threads in each block.


Below, in Section~\ref{subsec:Approach-ParMERLIN}, we introduce general architecture and data structures of \ParMERLIN{}, the parallel algorithm to discover discords of every possible length in a specified range on GPU that is based on the original serial MERLIN algorithm~\cite{DBLP:conf/icdm/NakamuraIMK20}. \ParMERLIN{} employs \ParDRAG{} (\ParDRAGtranscript{})~\cite{KraevaZ23}, our parallel version of the original serial DRAG algorithm~\cite{DBLP:conf/icdm/YankovKR07}. Similarly to the original serial algorithm, \ParDRAG{} performs in two phases, where each phase is parallelized separately from the other. In Sections~\ref{subsec:Approach-ParDRAG-Selection} and~\ref{subsec:Approach-ParDRAG-Refinement}, we discuss the parallelization of the candidate selection and discords refinement phases, respectively.

\input{parmerlin-parmerlin.tex}

\input{parmerlin-pardrag.tex}

%% file: parmerlin-parmerlin.tex

\subsection{General Architecture and Data Structures}
\label{subsec:Approach-ParMERLIN}

\subsubsection{Avoiding Redundant Calculations}
\label{subsubsec:Approach-ParMERLIN-Reduce}

Basically, our parallel algorithm follows the computational scheme of its serial predecessor (see Alg.~\ref{alg:MERLIN}). \ParMERLIN{} employs repeated calls of \ParDRAG{}, our designed parallel version of the serial DRAG algorithm (see Alg.~\ref{alg:DRAG} and lines~3, 9, 13, and 15 in Alg.~\ref{alg:MERLIN}) for a graphics processor. As opposed to the original algorithm, \ParMERLIN{} avoids redundant calculations in iterative calls of DRAG when the subsequence length is one more than at the previous step. Indeed, to calculate the distance between any two candidate subsequences, we need to partially repeat alike calculations regarding subsequences that are one less length self-matches to the candidates above. More formally, for any $i,\,j$ ($1 < i \leq n-m$ and $3 \leq m \ll n$), when calculating $\operatorname{ED^2_{norm}}(T_{i,\,m}, T_{j,\,m})$ through Equation~\ref{eq:EDnorm} from scratch, we partially repeat calculations of both mean values $\mu_{T_{i,\,m-1}}$ and $\mu_{T_{i,\,m-1}}$, and standard deviations $\sigma_{T_{i,\,m-1}}$ and $\sigma_{T_{i,\,m-1}}$ through Equation~\ref{eq:zNormalization}.

To avoid the overhead above, we employ the vectors $\bar{\mu}, \bar{\sigma} \in \mathds{R}^{n-minL+1}$, namely the mean values and standard deviations of all the given time series subsequences of the given length, respectively. In these vectors, first $n-m+1$ elements are processed, where $m$ is the given subsequence length ($\textit{minL} \leq m \leq \textit{maxL}$), and the rest ones are left unattended. 

For the $\textit{minL}$\nobreakdash-length subsequences, these vectors are calculated according to Equation~\ref{eq:zNormalization} once before the very first call of \ParDRAG{} (see line~3 in Alg.~\ref{alg:MERLIN}) whereas for the rest values of $m$, the vectors $\bar{\mu}$ and $\bar{\sigma}$ are updated before the each further call of \ParDRAG{} (see lines 9, 13, and 15 in Alg.~\ref{alg:MERLIN}) according to the following recurrent formulas:
\begin{eqnarray} 
	\label{eq:MuRecurrent}
	\mu_{T_{i,\,m+1}}&=&\dfrac{1}{m+1}\bigl(m\mu_{T_{i,\,m}} +t_{i+m}\bigr), \\
	\label{eq:SigmaRecurrent}	
	\sigma^2_{T_{i,\,m+1}}&=&\dfrac{m}{m+1} \Bigl( \sigma_{T_{i,\,m}}^2 + \dfrac{1}{m+1} \bigl(\mu_{T_{i,\,m}}-t_{i+m}\bigr)^2\Bigr).
\end{eqnarray}	
In order not to overload the article's text with details, we present the lemma with the proof of Equations~\ref{eq:MuRecurrent} and~\ref{eq:SigmaRecurrent} in \nameref{sec:Appendix}.

Initial calculation and update of the vectors $\bar{\mu}$ and $\bar{\sigma}$ are implemented as CUDA kernels. We form a grid of $N$ threads where the number of threads in each block is the algorithm's parameter that is set as a multiple of the GPU warp size. Each thread calculates elements of the vectors $\bar{\mu}$ and $\bar{\sigma}$ according to Equations~\ref{eq:MuRecurrent} and~\ref{eq:SigmaRecurrent}.

\subsubsection{Data Structures and Segmentation}
\label{subsubsec:Approach-ParMERLIN-DataStructures}

\begin{figure}[ht!]
	\centering
	\includegraphics[width=\linewidth]{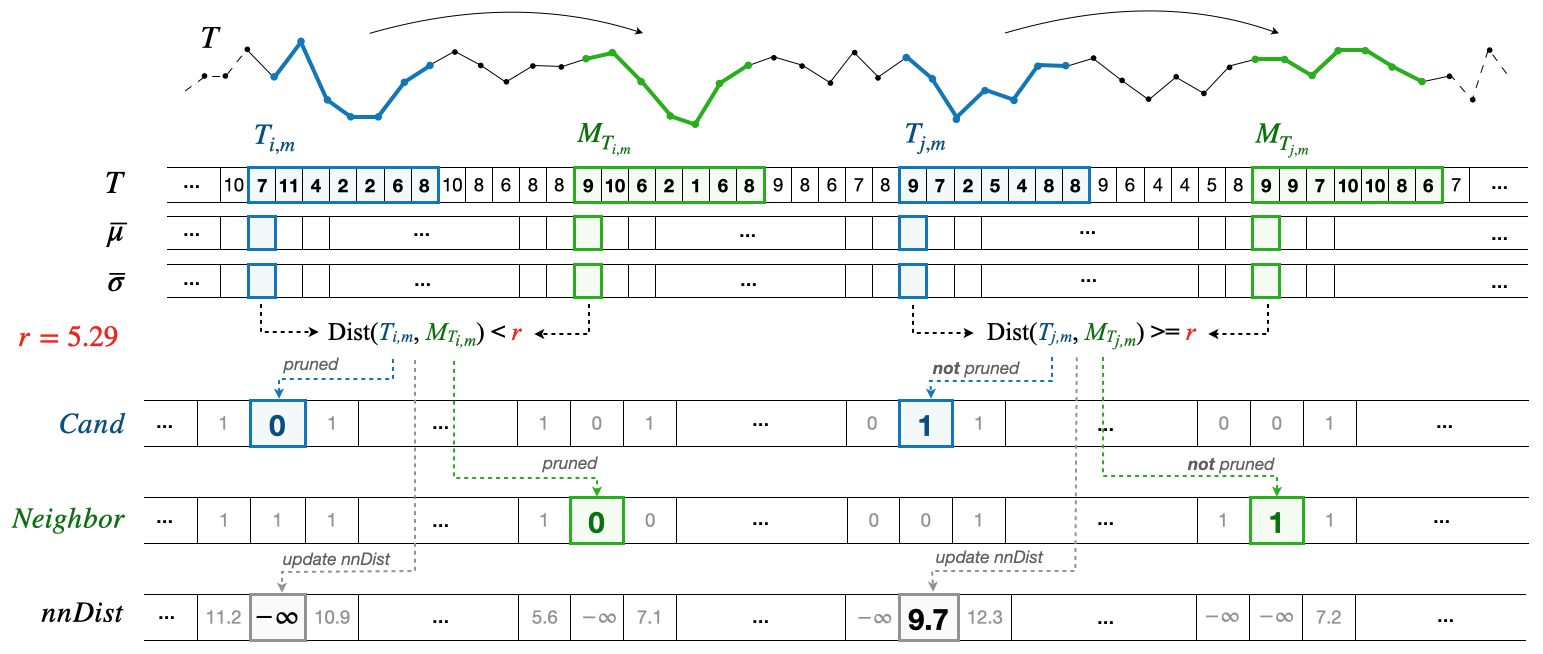}
	\caption{Data structures}
	\label{fig:DataStructures}
\end{figure}

In Fig.~\ref{fig:DataStructures}, we depict basic data structures involved in computations. The real-valued vector $\textit{nnDist} \in \mathds{R}^N$, for a subsequence of the input time series contains distance to its nearest neighbor: $\textit{nnDist}(i)=\min\bigl(\operatorname{ED^2_{norm}}(T_{i,\,m},M_{T_{i,\,m}})\bigr)$.
Two boolean-valued vectors $\textit{Cand}, \textit{Neighbor} \in \mathds{B}^N$ are the bitmaps for the subsequences and their nearest neighbors, respectively: $\textit{Cand}(i)=\texttt{TRUE}$ (or $\textit{Neighbor}(i)=\texttt{TRUE}$, respectively) if the subsequence $T_{i,\,m}$ (or its nearest neighbor, respectively) is a discord, and \texttt{FALSE} otherwise. These bitmaps are initialized with \texttt{TRUE} values. Further, we employ element-wise conjunction of the bitmaps above to discard more candidates during processing relying on the obvious fact that a subsequence that is not a discord cannot have a nearest neighbor that is a discord.

\begin{figure}[ht!]
	\centering
	\includegraphics[width=\linewidth]{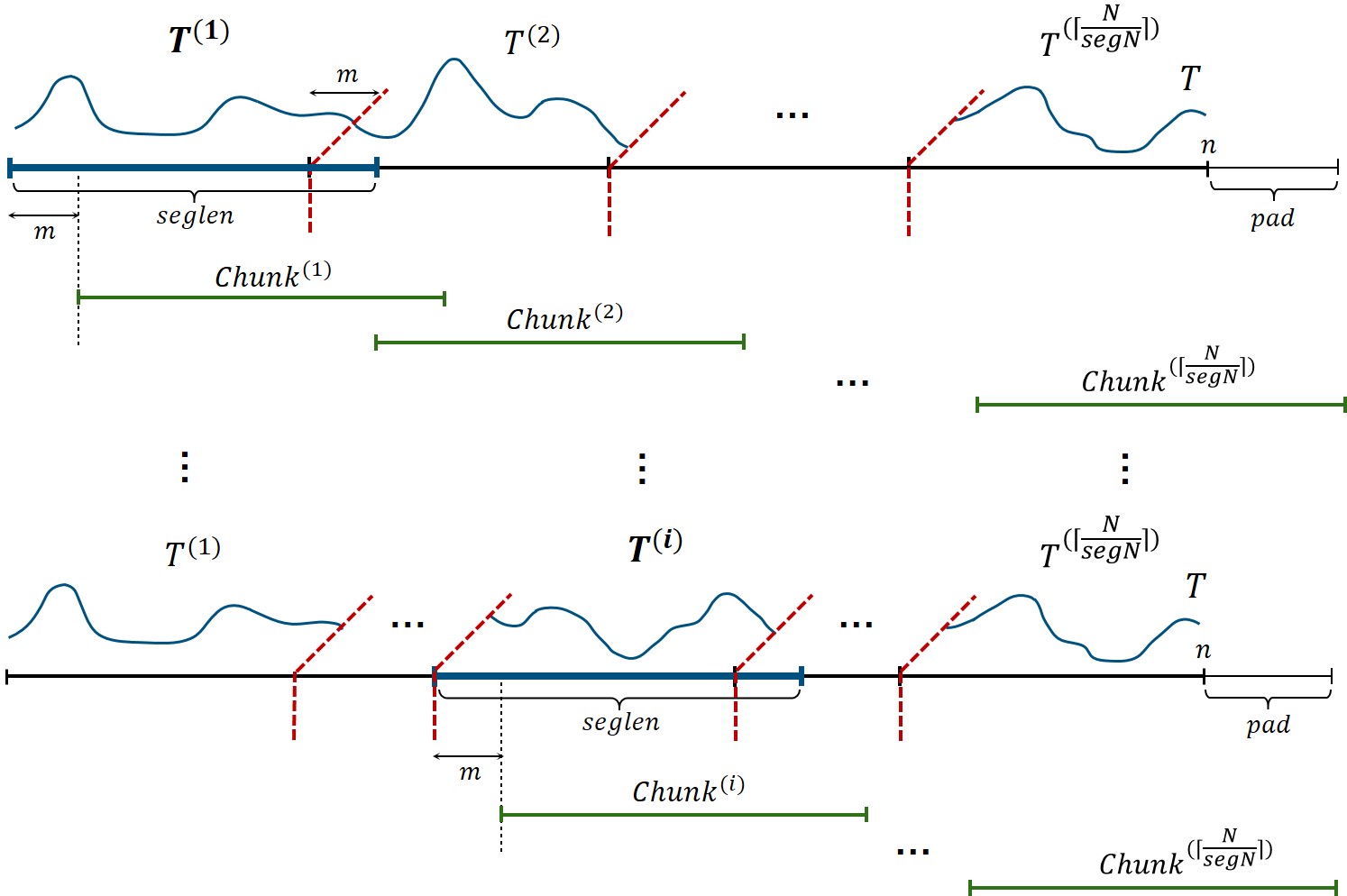}
	\caption{Data segmentation}
	\label{fig:DataSegmentation}
\end{figure}

To implement the candidate selection and discord refinement phases, we exploit the data parallelism concept and segment data as depicted in Fig.~\ref{fig:DataSegmentation}. The time series is divided into equal-length segments, where each segment is processed separately by a block of GPU threads. Performing the phase, the thread block scans the subsequences in chunk-wise manner. The number of elements in a chunk is equal to the segment length, and the first chunk begins with the $m$\nobreakdash-th element in the segment. Such a technique avoids redundant checks of the fact that  candidates and subsequences in chunks overlap.

The segment length is the algorithm's parameter to be set as a multiple of the GPU warp size. To balance the load of threads in the block, we require that the number of $m$\nobreakdash-length subsequences in the time series is a multiple of the number of subsequences of the specified length in the segment. If this is not the case, we pad the time series right with dummy positive infinity-valued elements. Let us denote the segment length and the number of $m$\nobreakdash-length subsequences in the segment by $\textit{seglen}$ and $\textit{segN}$, respectively, then $\textit{segN}=\textit{seglen}-m+1$. Let the number of dummy elements in the  rightmost segment is denoted by $\textit{pad}$, then it is defined as follows:
\begin{equation}
	pad = 
	\begin{cases}
		m-1, &\text{$N$ \text{mod} $\textit{segN}$ = 0}\\
		\lceil \dfrac{N}{\textit{segN}} \rceil \cdotp \textit{segN} + 2(m-1) - n, &\text{otherwise}
	\end{cases}.
\end{equation}

%% file: parmerlin-pardrag.tex
\subsection{Parallelization of Candidate Selection}
\label{subsec:Approach-ParDRAG-Selection}

\begin{algorithm}[!ht]
	\caption{\textsc{\ParDRAG{}select} (\textsc{in} $T$, $m$, $r$; \textsc{out} $\mathcal{C}$)}
	\begin{algorithmic}[1] 
		\State{$\textit{Cand} \leftarrow \overline{\texttt{TRUE}}$; $\textit{Neighbor} \leftarrow \overline{\texttt{TRUE}}$}
		\ForAll{$T^{(i)} \in T$}\Comment{\texttt{PARALLEL (block)}}			
			\ForAll{$\textit{Chunk}^{(j)} \in T^{(i)}$ \textbf{where} $i \leq j$}\Comment{\texttt{PARALLEL (thread)}}
				\If{$i=j$}
					\State{$\textit{QTrow} \leftarrow \textsc{CalcDotProducts}(T_{1,\,m}^{(i)}, \textit{Chunk}^{(j)})$}
					\State{\textbf{continue}}	
				\EndIf	
				\State{$\textit{QTrow} \leftarrow \textsc{UpdateDotProducts}(\textit{QTrow}, T_{1,\,m}^{(i)}, \textit{Chunk}^{(j)})$}
				\State{$\textit{QTcol} \leftarrow \textsc{CalcDotProducts}(\textit{Chunk}_{1,\,m}^{(j)}, T^{(i)})$}
				\State{$\textit{dist} \leftarrow \textsc{CalcDist}(\textit{Chunk}_{1,\,m}^{(j)}, T^{(i)}, \textit{QTcol}, \bar{\mu}, \bar{\sigma})$}	
				\If{$\textit{dist} < r$}
					\State{$\textit{Cand}(i\cdotp \textit{segN}+tid) \leftarrow \texttt{FALSE}; \textit{Neighbor}(j\cdotp \textit{segN}+1) \leftarrow \texttt{FALSE}$}
				\Else
					\State{$\textit{nnDist}(j \cdot \textit{segN}+1) \leftarrow  \min \bigl(\textit{dist},\textit{nnDist}(j \cdot \textit{segN}+1) \bigr)$}	
				\EndIf
				\If{\textbf{not} $\bigvee_{k=i\cdotp \textit{segN}}^{(i+1)\cdotp \textit{segN}} \textit{Cand}(k)$} 
					\State{\textbf{break}}
				\EndIf
				\ForAll{$\textit{Chunk}_{k,\,m}^{(j)} \in S^m_{\textit{Chunk}^{(j)}}\smallsetminus \textit{Chunk}_{1,\,m}^{(j)}$}\Comment{\texttt{PARALLEL (thread)}}
					\State{$\textit{QTcol} \leftarrow \textsc{UpdateDotProducts}(\textit{QTcol}, \textit{QTrow}, \textit{Chunk}_{k,\,m}^{(j)}, T^{(i)})$}
					\State{$\textit{dist} \leftarrow \textsc{CalcDist}(\textit{Chunk}_{k,\,m}^{(j)}, T^{(i)}, \textit{QTcol}, \bar{\mu}, \bar{\sigma})$}
					\If{$\textit{dist} < r$}
						\State{$\textit{Cand}(i\cdotp \textit{segN}+tid) \leftarrow \texttt{FALSE}; \textit{Neighbor}(j\cdotp \textit{segN}+k) \leftarrow \texttt{FALSE}$}
					\Else
						\State{$\textit{nnDist}(j \cdot \textit{segN}+1) \leftarrow  \min \bigl(\textit{dist},\textit{nnDist}(j \cdot \textit{segN}+1) \bigr)$}	
					\EndIf
				\EndFor
				\If{\textbf{not} $\bigvee_{k=i\cdotp \textit{segN}}^{(i+1)\cdotp \textit{segN}} \textit{Cand}(k)$}
					\State{\textbf{break}}
				\EndIf		
			\EndFor
		\EndFor	
		\State{$\mathcal{C} \leftarrow \bigl\{\{T_{i,\,m} \in S^m_T;\, \textit{nnDist}(i)\} \mid 1 \leq i \leq n-m+1, \textit{Cand}(i) = \texttt{TRUE}\bigr\} $}
		\State\Return{$\mathcal{C}$}
	\end{algorithmic}
	\label{alg:ParSelect}
\end{algorithm}

Alg.~\ref{alg:ParSelect} depicts a pseudo code of our parallelization of the candidate selection phase. The respective CUDA kernel forms a grid consisting of $\lceil \tfrac{N}{\textit{segN}} \rceil$ blocks of $\textit{segN}$ threads in each block. 
A thread block considers the segment subsequences as local candidates to discords and performs chunk-wise processing of the subsequences that are located to the \emph{right} of the segment and do not overlap with the candidates. The processing of the subsequences is as follows. If the distance from the candidate to the subsequence is less than the parameter~$r$, then the candidate and the subsequence are excluded from further processing as obviously not discords (the corresponding flags in the bitmaps are set to \texttt{FALSE}). If all the local candidates are discarded, the block terminates all its threads ahead of schedule.

\begin{figure}[ht!]
	\centering
	\includegraphics[width=\linewidth]{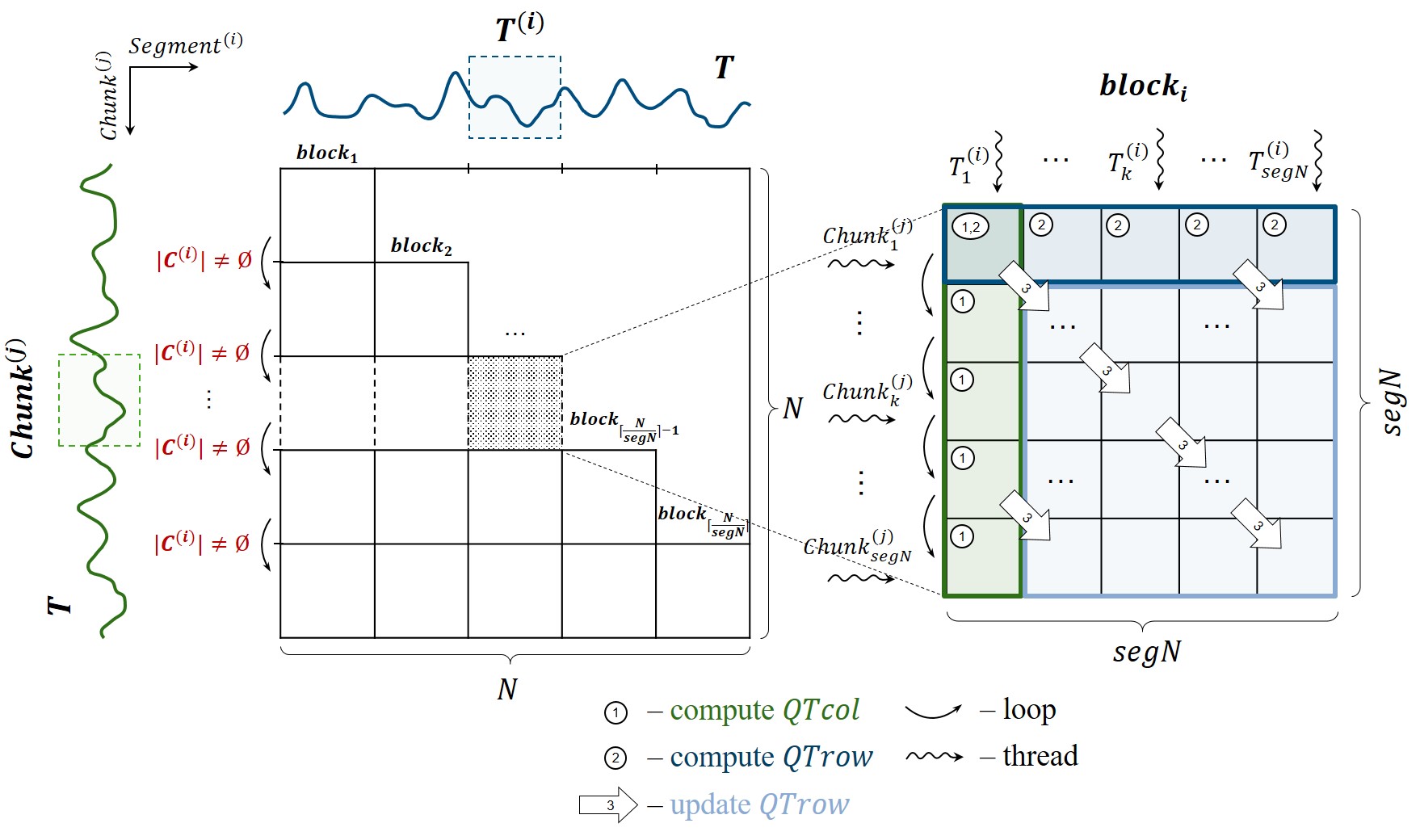}
	\caption{Computational kernel for the candidate selection phase}
	\label{fig:ParSelect}
\end{figure}

In Fig.~\ref{fig:ParSelect}, we show in detail how the block's threads work.
The thread block loads its segment into the shared memory once before starting calculations, and at each scanning step also loads there the current chunk located to the right of the segment. This technique allows for increasing the algorithm performance through reduction the number of reads the time series elements from the global memory. Next, the block threads calculate scalar products, storing the results in shared memory: firstly, products between the first subsequence of the segment and all the subsequences of the current chunk, and then these ones between the first subsequence of the current chunk and all the subsequences of the segment (the vectors $\textit{QTrow}, \textit{QTcol} \in \mathds{R}^{\textit{segN}}$ in lines~4--7 and~8 in Alg.~\ref{alg:ParSelect}, respectively).

Further, based on the obtained vector $\textit{QTcol}$ and the pre-calculated vectors $\bar{\mu}, \bar{\sigma}$, we calculate the distances between the first subsequence of the chunk and all the subsequences of the segment through the Equation~\ref{eq:EDnorm} (see line~9 in Alg.~\ref{alg:ParSelect}). Employing the calculated distances, we discard unpromising candidates located in the segment and the current chunk (see lines~10--11 in Alg.~\ref{alg:ParSelect}). If all the candidates in the segment are discarded, then the block stops (see lines~14--15 in Alg.~\ref{alg:ParSelect}).

After that, the block threads perform similar actions over the remaining subsequences of the current chunk, however, calculating scalar products more efficiently (see lines~16--24 in Alg.~\ref{alg:ParSelect}). We calculate the scalar products between the current subsequence of the chunk and all the subsequences of the segment (i.e., the vector $\textit{QTcol}$) based on the previously calculated vector $\textit{QTrow}$ and the vector $\textit{QTcol}$ obtained at the previous iteration (see line~16 in Alg.~\ref{alg:ParSelect}). 
To calculate the scalar product of the $k$\nobreakdash-th $(1<k \leq \textit{segN})$ subsequence in the $\textit{Chunk}^{(j)}$ and a subsequence in the segment $T^{(i)}$, we employ the following formula:
\begin{equation}
	\label{eq:QTcol}
	QTcol(tid) = 
	\begin{cases}
		\textit{QTcol}(\textit{tid}-1) + T_{\textit{tid},\,m}^{(i)} \cdotp \textit{Chunk}_{k,\,m}^{(j)}(m) -\\
			 \quad -T_{\textit{tid}-1,\,m}^{(i)}(1) \cdotp \textit{Chunk}_{k-1,\,m}^{(j)}(1),  & \text{$1 < \textit{tid} \leq \textit{segN}$}
	\\
		\textit{QTrow}(k), & \text{$\textit{tid} = 1$}
	\end{cases},
\end{equation}
where $tid$~denotes the number of a thread in the block. Since in Equation~\ref{eq:QTcol}, the first term is obtained from the previous iteration, we achieve the $O(1)$ complexity of the scalar product calculation instead of $O(m)$ as in the straightforward case.

\subsection{Parallelization of Discord Refinement}
\label{subsec:Approach-ParDRAG-Refinement}

\begin{algorithm}[!ht]
	\caption{\textsc{\ParDRAG{}refine} (\textsc{in} $T$, $m$, $r$; \textsc{out} $\mathcal{D}$)}
\begin{algorithmic}[1]
	\ForAll{$T_{i,\,m} \in S^m_T$}\Comment{\texttt{PARALLEL (thread)}}
		\State{$\textit{Cand}(i) \leftarrow \textit{Cand}(i) \wedge \textit{Neighbor}(i)$}
	\EndFor
	\ForAll{$T^{(i)} \in T$ \textbf{where} $\bigwedge_{k=i \cdotp \textit{segN}}^{(i+1) \cdotp \textit{segN}} \textit{Cand}(k) = \texttt{TRUE}$}\Comment{\texttt{PARALLEL (block)}}
		\ForAll{$Chunk^{(j)} \in T^{(i)}$ \textbf{where} $i \geq j$}\Comment{\texttt{PARALLEL (thread)}}
			\If{$i=j$}
				\State{$\textit{QTrow} \leftarrow \textsc{CalcDotProducts}(T_{1,\,m}^{(i)}, \textit{Chunk}^{(j)})$}
				\State{\textbf{continue}}	
			\EndIf	
			\State{$\textit{QTrow} \leftarrow \textsc{UpdateDotProducts}(\textit{QTrow}, T_{1,\,m}^{(i)}, \textit{Chunk}^{(j)})$}
			\State{$\textit{QTcol} \leftarrow \textsc{CalcDotProducts}(\textit{Chunk}_{1,\,m}^{(j)}, T^{(i)})$}
			\State{$dist \leftarrow \textsc{CalcDist}(\textit{Chunk}_{1,\,m}^{(j)}, T^{(i)}, \textit{QTcol}, \bar{\mu}, \bar{\sigma})$}	
			\If{$\textit{dist} < r$}
				\State{$\textit{Cand}(i \cdotp \textit{segN}+\textit{tid}) \leftarrow \texttt{FALSE}$}	
			\Else
				\State{$\textit{nnDist}(j \cdot \textit{segN}+\textit{tid}) \leftarrow \min\bigl(\textit{dist},\textit{nnDist}(j \cdot \textit{segN}+\textit{tid})\bigr)$}		
			\EndIf
			\If{\textbf{not} $\bigvee_{k=i \cdotp \textit{segN}}^{(i+1) \cdotp \textit{segN}} \textit{Cand}(k)$}
				\State{\textbf{break}}
			\EndIf
			\ForAll{$\textit{Chunk}_{k,m}^{(j)} \in S^m_{\textit{Chunk}^{(j)}} \smallsetminus \textit{Chunk}_{1,\,m}^{(j)}$}\Comment{\texttt{PARALLEL (thread)}}
				\State{$\textit{QTcol} \leftarrow \textsc{UpdateDotProducts}(\textit{QTcol}, QTrow, Chunk_{k,\,m}^{(j)}, T^{(i)})$}
				\State{$\textit{dist} \leftarrow \textsc{CalcDist}(\textit{Chunk}_{k,\,m}^{(j)}, T^{(i)}, \textit{QTcol}, \bar{\mu}, \bar{\sigma})$}
				\If{$dist < r$}
					\State{$\textit{Cand}(i \cdotp \textit{segN}+\textit{tid}) \leftarrow \texttt{FALSE}$}
				\Else
					\State{$\textit{nnDist}(j \cdot segN+tid) \leftarrow \min\bigl(\textit{dist},\textit{nnDist}(j \cdot \textit{segN}+\textit{tid})\bigr)$}			
				\EndIf
			\EndFor
			\If{\textbf{not} $\bigvee_{k=i \cdotp \textit{segN}}^{(i+1) \cdotp \textit{segN}} \textit{Cand}(k)$}
				\State{\textbf{break}}
			\EndIf		
		\EndFor
	\EndFor
	\State{$\mathcal{D} \leftarrow \bigl\{\{T_{i,\,m} \in S^m_T;\, \textit{nnDist}(i)\} \mid 1 \leq i \leq n-m+1, \textit{Cand}(i) = \texttt{TRUE}\bigr\} $}
	\State\Return{$\mathcal{D}$}
\end{algorithmic}
\label{alg:ParRefine}
\end{algorithm}

The discord refinement phase (see Alg.~\ref{alg:ParRefine}) parallelization is implemented through two CUDA kernels called one after the other. The first one trivially refines discords obtained in the previous phase through the element-wise conjunction of the $Cand$ and $Neighbor$ bitmap vectors, writing the result to the former. This operation allows for pruning the nearest neighbors of the subsequences discarded at the selection phase.

The second kernel performs non-trivial refinement and is parallelized similar to the selection phase involving only those segments of the time series whose set of local candidates is not empty (see line~3 in Alg.~\ref{alg:ParRefine}). The algorithm scans and processes the subsequences that do not overlap with the candidates and are located to the \emph{left} of the segment (see line~4 in Alg.~\ref{alg:ParRefine}). If the distance between the candidate and the subsequence is less than the parameter $r$, then the candidate is discarded as obvious false positive.

%% file: parmerlin-experiments.tex

\section{Experimental Evaluation}
\label{sec:Experiments}

To evaluate the proposed algorithm, we carried out the experiments and study the performance of \ParMERLIN{} over various real-world and synthetic time series in comparison with analogs, and investigated the algorithm's scalability.
We designed the experiments to be easily reproducible with our  repository~\cite{KraevaZ22} that contains the algorithm's source code and all the datasets used in this work.
Below, Section~\ref{subsec:Experiments-Setup} describes hardware and time series employed in the experiments, and Section~\ref{subsec:Experiments-Results} presents experimental results and discussion.

\subsection{The Experimental Setup}
\label{subsec:Experiments-Setup}

\begin{table}[ht!]
	\caption{Time series employed in the experiments}
	\label{tab:Datasets}
	\centering
	\begin{tabular}{|l|c|c|l|}
		\hline
		\multicolumn{1}{|c|}{Time series} & \begin{tabular}[c]{@{}c@{}}Length \\ $(n)$\end{tabular}  & \begin{tabular}[c]{@{}c@{}}Discord length \\ $(\textit{minL}=\textit{maxL})$\end{tabular} & \multicolumn{1}{c|}{Domain} \\ 
		\hline
		Space shuttle & 50~000  & 150 & \begin{tabular}[l]{@{}l@{}}Measurements of a sensor on the NASA \\ spacecraft\end{tabular} \\ 
		\hline
		ECG & 45~000 & 200 & \multirow{3}{*}{Electrocardiogram of an adult patient} \\ 
		\cline{1-3}
		ECG\nobreakdash-2 & 21~600  & 400 & \\ 
		\cline{1-3}
		Koski\nobreakdash-ECG & 100~000 & 458 & \\ 
		\hline
		Respiration & 24~125 & 250 & Human breathing by chest expansion \\
		\hline
		Power demand & 33~220  & 750 & Annual energy consumption of an office \\ 
		\hline
		RandomWalk1M & $10^7$ & 512 & \multirow{2}{*}{Synthetic time series} \\ 
		\cline{1-3}
		RandomWalk2M & $2 \cdot 10^7$ & 512 & \\ 
		\hline
	\end{tabular}
\end{table}

In our study, we employed the time series listed in Tab.~\ref{tab:Datasets} that are also used in the experimental evaluation of HOTSAX~\cite{DBLP:conf/icdm/KeoghLF05}, KBF\_GPU~\cite{ThuyAC21}, and Zhu \textit{et al.}'s~\cite{DBLP:journals/tpds/ZhuJGD21} algorithm. The Space shuttle data~\cite{NASAShuttleValveData05} are solenoid current measurements on a Marotta MPV\nobreakdash-41 series valve as the valve is cycled on and off under various test conditions in a laboratory where the valves are used to control fuel flow on the NASA spacecraft. The ECG and ECG2~\cite{PhysioBank00}, and Koski\nobreakdash-ECG~\cite{DBLP:journals/prl/Koski96} time series are electrocardiograms of adult patients. The Respiration time series~\cite{DBLP:conf/icdm/KeoghLF05} shows a patient’s breathing (measured by thorax extension), as s/he wakes up. The Power demand time series reflects the energy consumption of the research center in Netherlands for 1997~\cite{DBLP:conf/infovis/WijkS99}. The RandomWalk1M and RandomWalk2M are our generated time series through the random walk model~\cite{Pearson1905}.

\begin{table}[!ht]
	\caption{Hardware platform of the experiments}
	\label{tab:Hardware}
	\centering
	\begin{tabular}{|l|c|c|}
		\hline           
		\multicolumn{1}{|c|}{Specifications} & \multicolumn{1}{c|}{GPU\nobreakdash-SUSU} & GPU\nobreakdash-MSU \\ 
		\hline
		Brand and product line & \multicolumn{2}{c|}{NVIDIA Tesla} \\ 
		\hline
		Model & \multicolumn{1}{c|}{V100} & P100 \\ 
		\hline
		\# cores & \multicolumn{1}{c|}{5~120} & 3~584 \\ 
		\hline
		Core frequency, GHz & \multicolumn{1}{c|}{1.3} &  1.19\\ 
		\hline
		Memory, Gb & \multicolumn{1}{c|}{32} &  16\\ 
		\hline
		\begin{tabular}[l]{@{}l@{}}Peak performance\\ 
			(double precision), TFLOPS
		\end{tabular} & \multicolumn{1}{c|}{7} & 4 \\ 
		\hline
	\end{tabular}
\end{table}

Tab.~\ref{tab:Hardware} summarizes hardware platform of our experiments, where GPU\nobreakdash-SUSU and GPU\nobreakdash-MSU denote graphics processors installed in the HPC centers of the South Ural State University~\cite{DolganinaIBR22} and Moscow State University~\cite{DBLP:journals/superfri/VoevodinANSSSSV19}, respectively. 

\subsection{Results and Discussion}
\label{subsec:Experiments-Results}

\subsubsection{Comparison with Analogs}
\label{subsubsec:Experiments-Comparison}

In the experiments, we compared \ParMERLIN{} with two algorithms, namely, KBF\_GPU~\cite{ThuyAC21} and Zhu \textit{et~al.}'s~\cite{DBLP:journals/tpds/ZhuJGD21} since our thorough review of related work (see Section~\ref{sec:RelatedWork}) did not reveal other GPU-oriented parallel competitors. Since the authors of the above rivals do not provide their source codes, for a fair comparison, in the experiments, we utilize time series and hardware identical to those employed in~\cite{ThuyAC21} and~\cite{DBLP:journals/tpds/ZhuJGD21}, respectively, and compare our results with ones reported in the original papers by the authors. The time series Koski\nobreakdash-ECG (see Tab.~\ref{tab:Datasets}) was employed to compare \ParMERLIN{} with KBF\_GPU on the GPU\nobreakdash-SUSU system (see Tab.~\ref{tab:Hardware}), and the rest time series were involved in comparison with Zhu \textit{et~al.}'s algorithm on GPU\nobreakdash-MSU. 
We omit the MERLIN performance results since, as expected, the original serial algorithm is significantly inferior to its parallel descendant, although in the experiments, we confirmed that \ParMERLIN{} produces exactly the same results as MERLIN. For each experiment, we ran \ParMERLIN{} 10~times and took the average value as the final running time.

Since the rival algorithms discover only the top\nobreakdash-1 discord whereas our algorithm finds all the discords of each length in a specified length range, to provide a fair comparison, in the experiments, we employ two following settings for \ParMERLIN{}. First, we set $\textit{minL}=\textit{maxL}$ for the range above. Second, we measure both running time of \ParMERLIN{} and the number of discords found to further show also the average time spent by our algorithm to discover one discord.


\input{plots/runtime-number-discords-vs-KBF-line-plot}

In Fig.~\ref{fig:Performance-KBF}, we compare the performance of \ParMERLIN{} and KBF\_GPU. It can be seen that our algorithm significantly outruns the rival in terms of both the overall running time and the average running time to discover one discord. Obviously, the reason is that KBF\_GPU implements brute-force approach whereas \ParMERLIN{} avoids redundant calculations and exploits advanced data structures.


\input{plots/runtime-number-discords-vs-Zhu}

Figure~\ref{fig:Performance-Zhu} depicts the experimental results on the \ParMERLIN{} performance comparing with Zhu \textit{et~al.}'s algorithm. It can be seen that Zhu \textit{et~al.}'s algorithm significantly outruns \ParMERLIN{}: up to 20~times and up to two orders of magnitude greater over real and synthetic time series, respectively. However, at the same time,  \ParMERLIN{} discovers substantially more discords: at least to two and seven orders of magnitude greater over real and synthetic time series, respectively. Thus, comparing the average running time to discover one discord, it can be seen that \ParMERLIN{} significantly outruns the rival starting at least from the moment when we set $topK$, the number of discords to be discovered, as a quarter of actual number of discords found: at least two times and three orders of magnitude over real and synthetic time series, respectively. 

\subsubsection{Scalability of PALMAD}
\label{subsubsec:Experiments-Scalability}


In addition to comparison of our algorithm with analogs, we also study the scalability of \ParMERLIN{}. First, we investigate the impact of the segment length (the parameter $seglen$, see Section~\ref{subsec:Approach-ParDRAG-Selection}) on the \ParMERLIN{} performance. Second, we asses our algorithm's performance depending on two input parameters that directly affect the amount of calculations, namely the time series length and discord range length. 

\input{plots/avgruntime-wrt-seglen}

In Fig.~\ref{fig:Scalability}, we show experimental results regarding the impact of the segment length on the \ParMERLIN{} performance. It can be seen that the algorithm's running time is proportional to the segment length for both real-world and synthetic time series, and the greater value of the segment length provides higher performance. This can be explained by the fact that when the segment length increases, the overhead of reading and writing segments in the GPU shared memory decreases. Moreover, this was the reason that we took $seglen=512$ in the above-described experiments. 

\input{plots/avgruntime-wrt-ts-lengths-line-plot}

\input{plots/avgruntime-wrt-discord-length-ranges-line-plot}

In Fig.~\ref{fig:Scalability-TimeSeriesLength} and Fig.~\ref{fig:Scalability-DiscordLengthRange}, we depict the performance of our algorithm depending on the time series length and  on the discord length range, respectively, for the cases of real-world and synthetic data. It can be observed that the algorithm's running time is proportional to the above-mentioned parameters for both real-world and synthetic time series.


%% file: plots/runtime-number-discords-vs-KBF-line-plot.tex
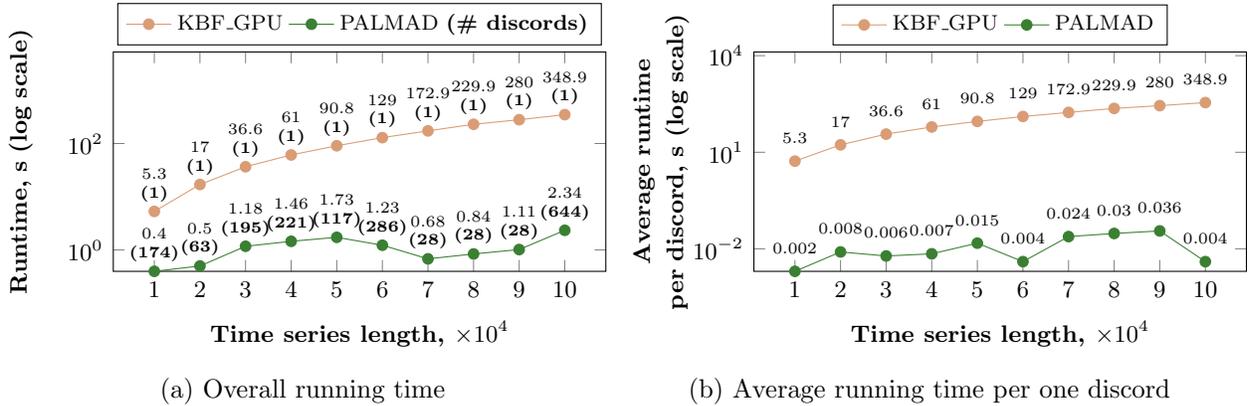
\begin{figure}[ht!]
	\centering
	\begin{subfigure}[h]{0.49\textwidth}	
		\centering	
		\pgfplotsset{every axis legend/.append style={
				at={(0.5,1.03)},
				anchor=south,
				legend cell align=left,
				font=\scriptsize}}
		\begin{tikzpicture}
			\begin{axis} [
			ymode=log,
			font=\footnotesize,
			log origin=infty,
			nodes near coords style={above},
			every node near coord/.append style={
				align=center,
				yshift=0pt,
				font=\tiny,
				color=black,
				/pgf/number format/fixed,
			},
			legend columns=-1,
			height=4.5cm,
			width=\textwidth,
			xlabel={\textbf{Time series length, \(\times 10^4\)}},
			ylabel style = {text width=4.5cm, align=center},
			ylabel = {\textbf{Runtime, s (log scale)}},
			enlarge y limits = {value = .4, upper},
			xtick={1, 2, 3, 4, 5, 6, 7, 8, 9, 10},
			x tick label style={align=center},
			]
			
			\addplot[mark=*, Tan, mark options={fill=Tan}] coordinates {(1, 5.3) (2, 17) (3, 36.6) (4, 61) (5, 90.8) (6, 129) (7, 172.9) (8, 229.9) (9, 280) (10, 348.9)};
			
			\addplot[mark=*, OliveGreen, mark options={fill=OliveGreen}] coordinates {(1,0.397) (2,0.5) (3,1.173) (4,1.455) (5,1.729) (6,1.228) (7,0.681) (8,0.842) (9,1.018) (10,2.342)};
			
			\addplot[only marks, nodes near coords, point meta=explicit symbolic] coordinates { (1, 5.3) [{5.3 \\ \textbf{(1)}}] (2, 17) [17 \\ \textbf{(1)}] (3, 36.6) [36.6 \\ \textbf{(1)}] (4, 61) [61 \\ \textbf{(1)}] (5, 90.8) [90.8 \\ \textbf{(1)}] (6, 129) [129 \\ \textbf{(1)}] (7, 172.9) [172.9 \\ \textbf{(1)}] (8, 229.9) [229.9 \\ \textbf{(1)}] (9, 280) [280 \\ \textbf{(1)}] (10, 348.9) [348.9 \\ \textbf{(1)}]}; 
			
			\addplot[only marks, nodes near coords, point meta=explicit symbolic, xshift=0pt] coordinates {(1, 0.397) [0.4 \\ \textbf{(174)}] (2, 0.5) [0.5 \\ \textbf{(63)}] (3, 1.173) [1.18 \\ \textbf{(195)}] (4, 1.455) [1.46 \\ \textbf{(221)}] (5, 1.729) [1.73 \\ \textbf{(117)}] (6, 1.228) [1.23 \\ \textbf{(286)}] (7, 0.681) [0.68 \\ \textbf{(28)}] (8, 0.842) [0.84 \\ \textbf{(28)}] (9, 1.018) [1.11 \\ \textbf{(28)}] (10, 2.342) [2.34 \\ \textbf{(644)}]};
			
			\legend {KBF\_GPU, PALMAD \textbf{(\# discords)}};
			
			\end{axis}
		\end{tikzpicture}
		\caption{Overall running time}
		\label{subfig:Performance-KBF-Overall} 
	\end{subfigure}
	\hfill
	\begin{subfigure}[h]{0.49\textwidth}	
		\centering			
		\begin{tikzpicture}
			\pgfplotsset{every axis legend/.append style={
				at={(0.5,1.03)},
				anchor=south,
				legend cell align=left,
				font=\scriptsize}}
			\begin{axis} [
			ymode=log,
			font=\footnotesize,
			log origin=infty,
			nodes near coords style={above},
			every node near coord/.append style={
				align=center,
				yshift=3pt,
				font=\tiny,
				color=black
			},
			legend columns=-1,
			height=4.5cm,
			width=\textwidth,
			xlabel={\textbf{Time series length, \(\times 10^4\)}},
			ylabel style = {text width=4.5cm, align=center},
			ylabel={\textbf{Average runtime\\per discord, s  (log scale)}},
			enlarge y limits = {value = .30, upper},
			xtick={1, 2, 3, 4, 5, 6, 7, 8, 9, 10},
			x tick label style={align=center},
			]

			\addplot[mark=*, Tan, mark options={fill=Tan}] coordinates {(1, 5.3) (2, 17) (3, 36.6) (4, 61) (5, 90.8) (6, 129) (7, 172.9) (8, 229.9) (9, 280) (10, 348.9)};
			
			\addplot[mark=*, OliveGreen, mark options={fill=OliveGreen}] coordinates {(1, 0.002) (2, 0.008) (3, 0.006) (4, 0.007) (5, 0.015) (6, 0.004) (7, 0.024) (8, 0.03) (9, 0.036) (10, 0.004)};
			
			\addplot[only marks, nodes near coords, point meta=explicit symbolic] coordinates { (1, 5.3) [5.3] (2, 17) [17] (3, 36.6) [36.6] (4, 61) [61] (5, 90.8) [90.8] (6, 129) [129] (7, 172.9) [172.9] (8, 229.9) [229.9] (9, 280) [280] (10, 348.9) [348.9]}; 
			
			\addplot[only marks, nodes near coords, point meta=explicit symbolic, xshift=0pt] coordinates {(1, 0.002) [0.002] (2, 0.008) [0.008] (3, 0.006) [0.006] (4, 0.007) [0.007] (5, 0.015) [0.015] (6, 0.004) [0.004] (7, 0.024) [0.024] (8, 0.03) [0.03] (9, 0.036) [0.036] (10, 0.004) [0.004]};
			
			\legend {KBF\_GPU, PALMAD};
			
			\end{axis}
		\end{tikzpicture}
		\caption{Average running time per one discord}
		\label{subfig:Performance-KBF-Average} 
	\end{subfigure}
	\vspace{5pt}
	\caption{Performance of \ParMERLIN{} in comparison with KBF\_GPU}
	\label{fig:Performance-KBF}
\end{figure}

%% file: plots/runtime-number-discords-vs-Zhu.tex
\begin{figure}[ht!]
	\centering
	\begin{subfigure}{\textwidth}	
		\centering	
		\begin{tikzpicture}
			\begin{axis} [ybar=2pt,
			bar width = 11pt,
			xtick = data,
			ymin = 0,
			ymode=log,
			font=\footnotesize,
			log origin=infty,
			nodes near coords style={above},
			every node near coord/.append style={
				align=center,
				yshift=-3pt,
				font=\tiny
			},
			legend style={
				legend pos=north west,
				legend cell align=left,
				font=\scriptsize
			},
			legend image code/.code={
				\draw [#1] (0cm,-0.1cm) rectangle (0.15cm,0.15cm); },
			height=6cm,
			width=0.8\textwidth,
			ylabel style = {text width=6cm, align=center},
			ylabel = {\textbf{Runtime, s \\ (log scale)}},
			enlarge y limits = {value = .15, upper},
			enlarge x limits=0.1,
			x tick label style={rotate=0, align=center},
			symbolic x coords={Space shuttle, ECG, ECG2, Power demand, Respiration, RandomWalk1M, RandomWalk2M},
			xticklabels = {Space\\shuttle, ECG, ECG2, Power\\demand, Respiration, Random\\Walk1M, Random\\Walk2M}
			]
			
			\addplot coordinates {(Space shuttle, 7) (ECG, 76) (ECG2, 26) (Power demand, 63) (Respiration, 26) (RandomWalk1M, 5431) (RandomWalk2M, 15190)};
			
			\addplot[fill=green4,draw=black!90] coordinates {(Space shuttle, 71) (ECG, 1056) (ECG2, 495) (Power demand, 353) (Respiration, 584) (RandomWalk1M, 394838) (RandomWalk2M, 1568880)};
			
			\addplot[only marks, nodes near coords, point meta=explicit symbolic, xshift=0pt] coordinates { (Space shuttle, 7) [7 \\ \textbf{(1)}] (ECG, 76) [76 \\ \textbf{(1)}] (ECG2, 26) [26 \\ \textbf{(1)}] (Power demand, 63) [63 \\ \textbf{(1)}] (Respiration, 26) [26 \\ \textbf{(1)}] (RandomWalk1M, 5431) [5~431 \\ \textbf{(1)}] (RandomWalk2M, 15190) [15~190 \\ \textbf{(1)}]}; 
			
			\addplot[only marks, nodes near coords, point meta=explicit symbolic, xshift=0pt] coordinates { (Space shuttle, 71) [{71 \\ \textbf{(963)}}] (ECG, 1056) [1 056 \\ \textbf{(111)}] (ECG2, 495) [495 \\ \textbf{(1~618)}] (Power demand, 353) [353 \\ \textbf{(1~290)}] (Respiration, 584) [584 \\ \textbf{(7~915)}] (RandomWalk1M, 394 838) [394~838 \\ \textbf{(1~047~809)}] (RandomWalk2M, 1568880) [1 568~880 \\ \textbf{(2~096~385)}]}; 
			
			\legend {Zhu \emph{et al.}, PALMAD \textbf{(\# discords)}};
			
			\end{axis}
		\end{tikzpicture}
		\caption{Overall running time}
		\label{subfig:Performance-Zhu-Overall} 
	\end{subfigure}
	\hfill
	\begin{subfigure}{\textwidth}	
		\centering			
		\begin{tikzpicture}
			\begin{axis} [ybar=1pt,
			bar width = 6pt,
			xtick = data,
			ymode=log,
			font=\footnotesize,
			log origin=infty,
			nodes near coords align={vertical},  
			nodes near coords style={above},
			every node near coord/.append style={
				yshift=-2pt,
				font=\tiny,
				rotate=90,
				anchor=west
			},
			legend style={
				legend pos=north west,
				legend cell align=left,
				font=\scriptsize,
				legend columns=2, 
				column sep=0.2cm
			},
			legend image code/.code={
				\draw [#1] (0cm,-0.1cm) rectangle (0.15cm,0.15cm); },
			height=7cm,
			width=\textwidth,
			ylabel style={text width=7cm, align=center},
			ylabel={\textbf{Average runtime per discord, ms \\ (log scale)}},
			enlarge y limits = {value = .4, upper},
			x tick label style={rotate=0, align=center},
			symbolic x coords={Space shuttle, ECG, ECG2, Power demand, Respiration, RandomWalk1M, RandomWalk2M},
			xticklabels = {Space\\shuttle, ECG, ECG2, Power\\demand, Respiration, Random\\Walk1M, Random\\Walk2M}
			]
			
			\addplot coordinates {(Space shuttle, 6.5) (ECG, 75.56) (ECG2, 25.61) (Power demand, 63.37) (Respiration, 26.45) (RandomWalk1M, 5431.0) (RandomWalk2M, 15190.0)};
			
			\addplot[fill=green1,draw=black!90] coordinates {(Space shuttle, 71.24) (ECG, 1056.04) (ECG2, 495.45) (Power demand, 353.12) (Respiration, 584.4) (RandomWalk1M, 394837.75) (RandomWalk2M, 1568880.38)};
			
			\addplot[fill=green2,draw=black!90] coordinates {(Space shuttle, 0.74) (ECG, 95.14) (ECG2, 3.06) (Power demand, 2.74) (Respiration, 0.74) (RandomWalk1M, 3.77) (RandomWalk2M, 7.48)};
			
			\addplot[fill=green3,draw=black!90] coordinates {(Space shuttle, 0.3) (ECG, 38.06) (ECG2, 1.23) (Power demand, 1.09) (Respiration, 0.29) (RandomWalk1M, 1.51) (RandomWalk2M, 2.99)}; 
			
			\addplot[fill=green4,draw=black!90] coordinates {(Space shuttle, 0.15) (ECG, 19.03) (ECG2, 0.61) (Power demand, 0.55) (Respiration, 0.15) (RandomWalk1M, 0.75) (RandomWalk2M, 1.5)}; 
			
			\addplot[fill=green5,draw=black!90] coordinates {(Space shuttle, 0.1) (ECG, 12.69) (ECG2, 0.41) (Power demand, 0.36) (Respiration, 0.1) (RandomWalk1M, 0.5) (RandomWalk2M, 1.0)}; 
			
			\addplot[fill=green6,draw=black!90] coordinates {(Space shuttle, 0.07) (ECG, 9.51) (ECG2, 0.31) (Power demand, 0.27) (Respiration, 0.07) (RandomWalk1M, 0.38) (RandomWalk2M, 0.75)};
			
			\addplot[only marks, nodes near coords, point meta=explicit symbolic, xshift=0pt] coordinates { (Space shuttle, 6.5) [6.5] (ECG, 75.56) [75.56] (ECG2, 25.61) [25.61] (Power demand, 63.37) [63.37] (Respiration, 26.45) [26.45] (RandomWalk1M, 5431.0) [5~431.0] (RandomWalk2M, 15190.0) [15~190.0]}; 
			
			\addplot[only marks, nodes near coords, point meta=explicit symbolic, xshift=0pt] coordinates {(Space shuttle, 71.24) [71.24] (ECG, 1056.04) [1~056.04] (ECG2, 495.45) [495.45] (Power demand, 353.12) [353.12] (Respiration, 584.4) [584.4] (RandomWalk1M, 394837.75) [394~837.75] (RandomWalk2M, 1568880.38) [1~568~880.38]};

			\addplot[only marks, nodes near coords, point meta=explicit symbolic, xshift=0pt] coordinates {(Space shuttle, 0.74) [0.74] (ECG, 95.14) [95.14] (ECG2, 3.06) [3.06] (Power demand, 2.74) [2.74] (Respiration, 0.74) [0.74] (RandomWalk1M, 3.77) [3.77] (RandomWalk2M, 7.48) [7.48]};
			
			\addplot[only marks, nodes near coords, point meta=explicit symbolic, xshift=0pt] coordinates {(Space shuttle, 0.3) [0.3] (ECG, 38.06) [38.06] (ECG2, 1.23) [1.23] (Power demand, 1.09) [1.09] (Respiration, 0.29) [0.29] (RandomWalk1M, 1.51) [1.51] (RandomWalk2M, 2.99) [2.99]};
			
			\addplot[only marks, nodes near coords, point meta=explicit symbolic, xshift=0pt] coordinates {(Space shuttle, 0.15) [0.15] (ECG, 19.03) [19.03] (ECG2, 0.61) [0.61] (Power demand, 0.55) [0.55] (Respiration, 0.15) [0.15] (RandomWalk1M, 0.75) [0.75] (RandomWalk2M, 1.5) [1.5]};
			
			\addplot[only marks, nodes near coords, point meta=explicit symbolic, xshift=0pt] coordinates {(Space shuttle, 0.1) [0.1] (ECG, 12.69) [12.69] (ECG2, 0.41) [0.41] (Power demand, 0.36) [0.36] (Respiration, 0.1) [0.1] (RandomWalk1M, 0.5) [0.5] (RandomWalk2M, 1.0) [1.0]};
			
			\addplot[only marks, nodes near coords, point meta=explicit symbolic, xshift=0pt] coordinates {(Space shuttle, 0.07) [0.07] (ECG, 9.51) [9.51] (ECG2, 0.31) [0.31] (Power demand, 0.27) [0.27] (Respiration, 0.07) [0.07] (RandomWalk1M, 0.38) [0.38] (RandomWalk2M, 0.75) [0.75]};

			\legend {Zhu \emph{et al.}, {PALMAD, $\textit{topK}$ = 1}, {PALMAD, $\textit{topK}$ = 10\% of $|\mathcal{D}|$}, {PALMAD, $\textit{topK}$ = 25\% of $|\mathcal{D}|$}, {PALMAD, $topK$ = 50\% of $|\mathcal{D}|$}, {PALMAD, $\textit{topK}$ = 75\% of $|\mathcal{D}|$}, {PALMAD, $\textit{topK}$ = $|\mathcal{D}|$}};
			
			\end{axis}
		\end{tikzpicture}
		\caption{Average running time per one discord}
		\label{subfig:Performance-Zhu-Average} 
	\end{subfigure}
	\vspace{5pt}
	\caption{Performance of \ParMERLIN{} in comparison with Zhu \textit{et~al.}'s algorithm}
	\label{fig:Performance-Zhu}
\end{figure}
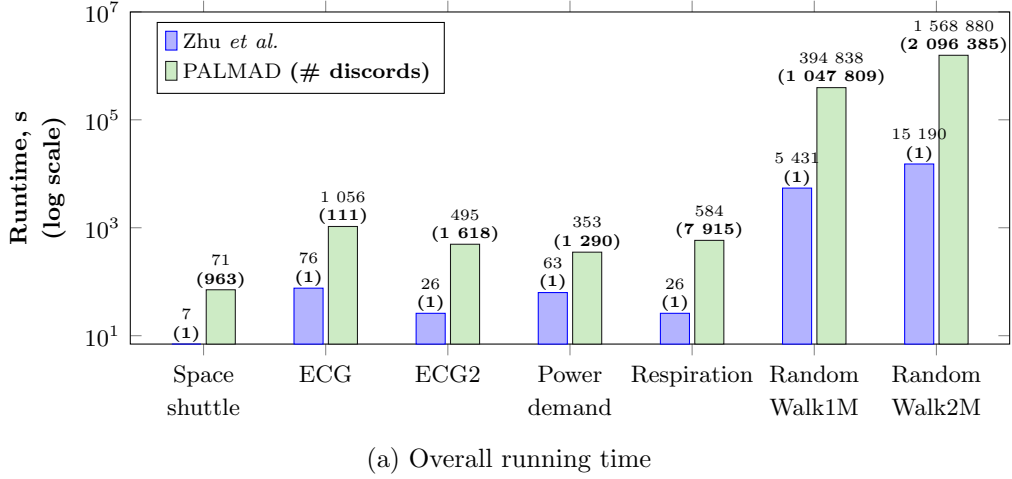
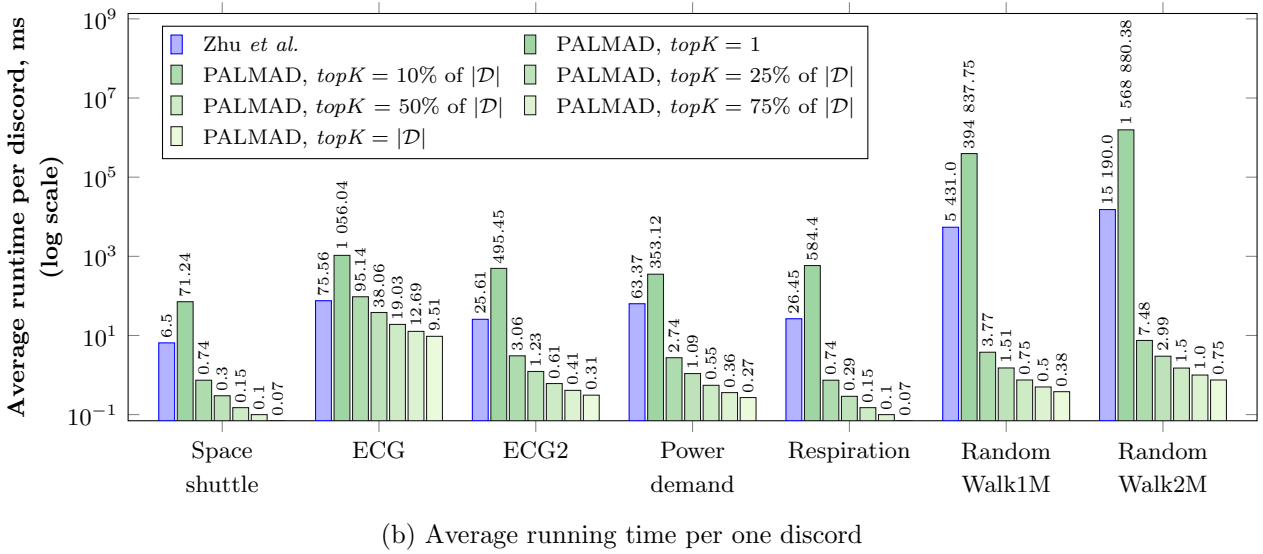

%% file: plots/avgruntime-wrt-seglen.tex
\begin{figure}[ht!]
	\centering
	\begin{subfigure}{0.6\textwidth}	
		\centering
		\begin{tikzpicture}
			\begin{axis} [ybar=2pt,
			bar width = 9pt,
			xtick = data,
			font=\footnotesize,
			log origin=infty,
			nodes near coords style={above},
			every node near coord/.append style={
				align=center,
				yshift=-1pt,
				font=\tiny
			},
			legend style={
				legend pos=north west,
				legend cell align=left,
				font=\scriptsize
			},
			legend image code/.code={
				\draw [#1] (0cm,-0.1cm) rectangle (0.15cm,0.15cm); },
			height=6.5cm,
			width=\textwidth,
			ylabel style = {text width=6.5cm, align=center},
			ylabel={\textbf{Average runtime per discord, ms}},
			enlarge y limits = {value = .1, upper},
			enlarge x limits = 0.15,
			x tick label style={rotate=0},
			symbolic x coords={Space shuttle, Respiration, ECG2, Power demand, ECG},
			x tick label style={align=center},
			xticklabels = {Space\\shuttle, Respiration, ECG2, Power\\demand, ECG}
			]
			
			\addplot[fill=green4,draw=black!90] coordinates {(Space shuttle, 45.2) (Respiration, 237.2) (ECG2, 311.9) (Power demand, 396.0) (ECG, 496.6)};
			
			\addplot[fill=green4,draw=black!90, postaction={pattern=north west lines}] coordinates {(Space shuttle, 36.9) (Respiration,196.8) (ECG2,233.0) (Power demand,262.1) (ECG,419.7)};
			
			\addplot[fill=green4,draw=black!90, postaction={pattern=horizontal lines}] coordinates {(Space shuttle, 33.8) (Respiration, 179.9) (ECG2, 197.9) (Power demand, 198.6) (ECG, 387.2)};
			
			\addplot[only marks, nodes near coords, point meta=explicit symbolic] coordinates { (Space shuttle, 45.2) [45.2] (Respiration, 237.2) [237.2] (ECG2, 311.9) [311.9] (Power demand, 396.0) [396.0] (ECG, 496.6) [496.6]}; 
			
			\addplot[only marks, nodes near coords, point meta=explicit symbolic, xshift=2pt] coordinates {(Space shuttle, 36.9) [36.9] (Respiration, 196.8) [196.8] (ECG2, 233.0) [233.0] (Power demand, 262.1) [262.1] (ECG, 419.7) [419.7]};
			
			\addplot[only marks, nodes near coords, point meta=explicit symbolic, xshift=2pt] coordinates {(Space shuttle, 33.8) [33.8] (Respiration, 179.9) [179.9] (ECG2, 197.9) [197.9] (Power demand, 198.6) [198.6] (ECG, 387.2) [387.2] };
			
			\legend{\textit{seglen} = 128, \textit{seglen} = 256, \textit{seglen} = 512};
			
			\end{axis}
		\end{tikzpicture}
	\end{subfigure}
	\hfill
	\begin{subfigure}{0.39\textwidth}	
		\centering			
		\begin{tikzpicture}
			\begin{axis} [ybar=2pt,
			bar width = 9pt,
			xtick = data,
			ymode=log,
			ytick={100000,200000,300000,400000,600000},
			yticklabels={1$\cdot 10^5$,2$\cdot 10^5$,3$\cdot 10^5$,4$\cdot 10^5$,6$\cdot 10^5$},
			font=\footnotesize,
			log origin=infty,
			nodes near coords style={above},
			every node near coord/.append style={
				yshift=-2pt,
				font=\tiny,
				rotate=90,
				anchor=west
			},
			legend style={
				legend pos=north west,
				legend cell align=left,
				font=\scriptsize
			},
			legend image code/.code={
				\draw [#1] (0cm,-0.1cm) rectangle (0.15cm,0.15cm); },
			height=6.5cm,
			width=0.8\textwidth,
			ylabel style = {text width=6.5cm, align=center},
			ylabel={\textbf{Average runtime per discord, ms \\ (log scale)}},
			enlarge y limits = {value = .3, upper},
			enlarge x limits = 0.55,
			x tick label style={rotate=0},
			symbolic x coords={RandomWalk1M, RandomWalk2M},
			x tick label style={align=center},
			xticklabels = {Random\\Walk1M, Random\\Walk2M}
			]
			
			\addplot[fill=green4,draw=black!90] coordinates {(RandomWalk1M, 110255.6) (RandomWalk2M, 430711.7)};
			
			\addplot[fill=green4,draw=black!90, postaction={pattern=north west lines}] coordinates {(RandomWalk1M, 108561.3) (RandomWalk2M, 428368.0)};
			
			\addplot[fill=green4,draw=black!90, postaction={pattern=horizontal lines}] coordinates {(RandomWalk1M, 107667.0) (RandomWalk2M, 425177.3)};
			
			\addplot[only marks, nodes near coords, point meta=explicit symbolic] coordinates { (RandomWalk1M, 110255.6) [110~255.6] (RandomWalk2M, 430711.7) [430~711.7]}; 
			
			\addplot[only marks, nodes near coords, point meta=explicit symbolic, xshift=0pt] coordinates {(RandomWalk1M, 108561.3) [108~561.3] (RandomWalk2M, 428368.0) [428~368.0]};
			
			\addplot[only marks, nodes near coords, point meta=explicit symbolic, xshift=0pt] coordinates {(RandomWalk1M, 107667.0) [107~667.0] (RandomWalk2M, 425177.3) [425~177.3]};
			
			
			\end{axis}
		\end{tikzpicture} 
	\end{subfigure}
	\caption{Scalability of the \ParMERLIN{} algorithm w.r.t. the segment length}
	\label{fig:Scalability}
\end{figure}

%% file: plots/avgruntime-wrt-ts-lengths-line-plot.tex
\begin{figure}[ht!]
	\centering
	\begin{subfigure}{\textwidth}	
		\begin{minipage}[t]{0.49\linewidth}
			\begin{tikzpicture}
				\begin{axis} 
				[
				font=\footnotesize,
				nodes near coords style={above},
				every node near coord/.append style={
					align=center,
					xshift=-5pt,
					yshift=0pt,
					font=\tiny,
					color=black
				},
				height=4.5cm,
				width=\textwidth,
				xlabel={\textbf{Time series length, \(\times 10^3\)}},
				ylabel={\textbf{Runtime, s}},
				ylabel style = {text width=4.5cm, align=center},
				enlarge y limits = {value = .2, upper},
				xtick={0, 5, 10, 25, 50, 75, 100},				
				x tick label style={/pgf/number format/.cd,%
					scaled x ticks = false,
					fixed,
					tick label style={rotate=0},
				},
				]
				
				\addplot[mark=*, OliveGreen, mark options={fill=OliveGreen}, nodes near coords] coordinates {(5, 26.87) (10, 52.70) (25, 131.71) (50, 210.82) (75, 257.05) (100, 324.24)};
				
				\end{axis}
			\end{tikzpicture}
		\end{minipage}
		\hfill
		\begin{minipage}[t]{0.49\linewidth}
			\begin{tikzpicture}
				\begin{axis} [
				font=\footnotesize,
				nodes near coords style={above},
				every node near coord/.append style={
					align=center,
					xshift=-5pt,
					font=\tiny,
					rotate=0,
					/pgf/number format/fixed,
					color=black
				},
				height=4.5cm,
				width=\textwidth,
				xlabel={\textbf{Time series length, \(\times 10^3\)}},
				ylabel style = {text width=4.5cm, align=center},
				ylabel={\textbf{Average runtime\\per discord, s}},
				enlarge y limits = {value = .2, upper},
				xtick={0, 5, 10, 25, 50, 75, 100},				
				x tick label style={/pgf/number format/.cd,%
					scaled x ticks = false,
					fixed,
					tick label style={rotate=0},
				},
				]
				
				\addplot[mark=*, OliveGreen, mark options={fill=OliveGreen}, nodes near coords] coordinates {(5, 0.06) (10, 0.12) (25, 0.29) (50, 0.46) (75, 0.56) (100, 0.71)};
				
				\end{axis}
			\end{tikzpicture} 
		\end{minipage}
		\caption{Real dataset (Koski-ECG, discord range is 458..916)}
	\end{subfigure}
	\hfill
	\begin{subfigure}{\textwidth}	
	\begin{minipage}[t]{0.49\linewidth}
		\begin{tikzpicture}
			\begin{axis} [
				font=\footnotesize,
				nodes near coords style={above},
				every node near coord/.append style={
					align=center,
					xshift=-5pt,
					yshift=0pt,
					font=\tiny,
					color=black,
					/pgf/number format/.cd,
					set thousands separator={\,},
					fixed
				},
				height=4.5cm,
				width=\textwidth,
				xlabel={\textbf{Time series length, \(\times 10^4\)}},
				ylabel={\textbf{Runtime, s}},
				ylabel style = {text width=4.5cm, align=center},
				y tick label style={/pgf/number format/.cd,%
					scaled y ticks = false,
					set thousands separator={\,},
					fixed},
				enlarge y limits = {value = .2, upper},
				xtick={0, 5, 10, 25, 50, 75, 100},				
				x tick label style={/pgf/number format/.cd,%
					scaled x ticks = false,
					fixed,
					tick label style={rotate=0},
				},
				]
				
				\addplot[mark=*, OliveGreen, mark options={fill=OliveGreen}, nodes near coords] coordinates {(5, 58.62) (10, 139.90) (25, 380.25) (50, 731.80) (75, 997.41) (100, 1376.78)};
				
			\end{axis}
		\end{tikzpicture}
	\end{minipage}
	\hfill
	\begin{minipage}[t]{0.49\linewidth}
		\begin{tikzpicture}
			\begin{axis} [
				font=\footnotesize,
				nodes near coords style={above},
				every node near coord/.append style={
					align=center,
					xshift=-5pt,
					font=\tiny,
					rotate=0,
					color=black
				},
				height=4.5cm,
				width=\textwidth,
				xlabel={\textbf{Time series length, \(\times 10^4\)}},
				ylabel style = {text width=4.5cm, align=center},
				ylabel={\textbf{Average runtime\\per discord, s}},
				enlarge y limits = {value = .2, upper},
				xtick={0, 5, 10, 25, 50, 75, 100},				
				x tick label style={/pgf/number format/.cd,%
					scaled x ticks = false,
					fixed,
					tick label style={rotate=0},
				},
				]
				
				\addplot[mark=*, OliveGreen, mark options={fill=OliveGreen}, nodes near coords] coordinates {(5, 0.46) (10, 1.09) (25, 2.97) (50, 5.72) (75, 7.79) (100, 10.76)};
				
			\end{axis}
		\end{tikzpicture} 
	\end{minipage}
	\caption{Synthetic dataset (RandomWalk1M, discord range is 128..256)}
	\end{subfigure}

	\vspace{5pt}
	\caption{Scalability of the \ParMERLIN{} algorithm w.r.t. the time series length}
	\label{fig:Scalability-TimeSeriesLength}
\end{figure}

%% file: plots/avgruntime-wrt-discord-length-ranges-line-plot.tex
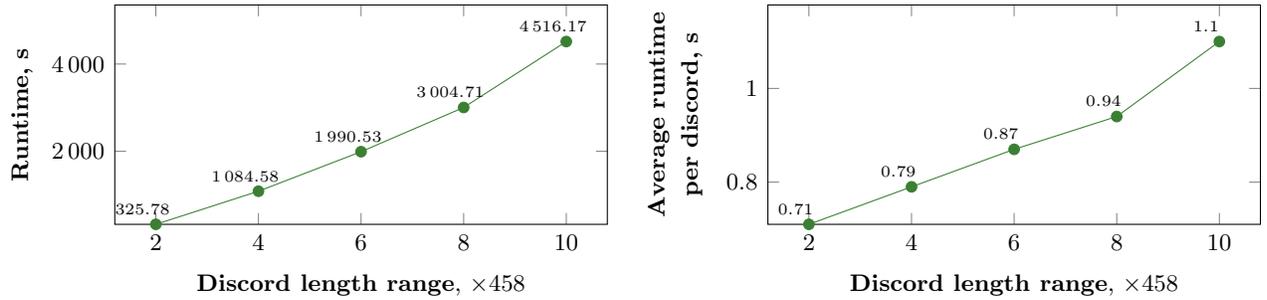
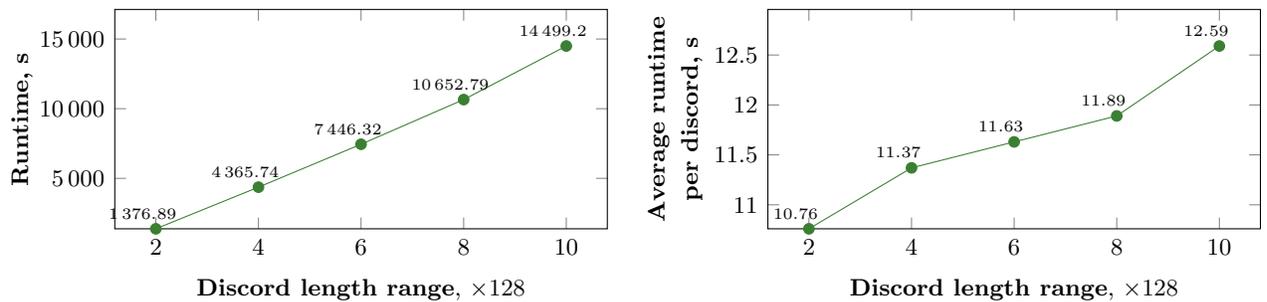
\begin{figure}[ht!]
	\centering
	\begin{subfigure}{\textwidth}	
	\begin{minipage}[t]{0.49\linewidth}
		\begin{tikzpicture}
			\begin{axis} 
			[
			font=\footnotesize,
			nodes near coords style={above},
			every node near coord/.append style={
				align=center,
				xshift=-5pt,
				yshift=0pt,
				font=\tiny,
				color=black,
				/pgf/number format/.cd,
				set thousands separator={\,},
				fixed,
			},
			height=4.5cm,
			width=\textwidth,
			xlabel={\textbf{Discord length range}, \(\times 458\)},
			ylabel={\textbf{Runtime, s}},
			ylabel style = {text width=4.5cm, align=center},
			y tick label style={/pgf/number format/.cd,%
				scaled y ticks = false,
				set thousands separator={\,},
				fixed},
			enlarge y limits = {value = .2, upper},
			]
			
			\addplot[mark=*, OliveGreen, mark options={fill=OliveGreen}, nodes near coords] coordinates {(2, 325.78) (4, 1084.58) (6, 1990.53) (8, 3004.71) (10, 4516.17)};
			
			\end{axis}
		\end{tikzpicture}
		\end{minipage}
		\hfill
		\begin{minipage}[t]{0.49\linewidth}
			\begin{tikzpicture}
				\begin{axis} [
				font=\footnotesize,
				nodes near coords style={above},
				every node near coord/.append style={
					align=center,
					xshift=-5pt,
					font=\tiny,
					rotate=0,
					color=black
				},
				height=4.5cm,
				width=\textwidth,
				xlabel={\textbf{Discord length range}, \(\times 458\)},
				ylabel style = {text width=4.5cm, align=center},
				ylabel={\textbf{Average runtime\\per discord, s}},
				enlarge y limits = {value = .2, upper},
				]
				
				\addplot[mark=*, OliveGreen, mark options={fill=OliveGreen}, nodes near coords] coordinates {(2, 0.71) (4, 0.79) (6, 0.87) (8, 0.94) (10, 1.1)};
				
				\end{axis}
			\end{tikzpicture} 
		\end{minipage}
	\caption{Real dataset (Koski-ECG)}
	\end{subfigure}
	\hfill
	\begin{subfigure}{\textwidth}	
		\centering	
		\begin{minipage}[t]{0.49\linewidth}
		\begin{tikzpicture}
			\begin{axis} [
			font=\footnotesize,
			nodes near coords style={above},
			every node near coord/.append style={
				align=center,
				xshift=-5pt,
				yshift=0pt,
				font=\tiny,
				color=black,
				/pgf/number format/.cd,
				set thousands separator={\,},
				fixed,
			},
			height=4.5cm,
			width=\textwidth,
			xlabel={\textbf{Discord length range}, \(\times 128\)},
			ylabel={\textbf{Runtime, s}},
			ylabel style = {text width=4.5cm, align=center},
			y tick label style={/pgf/number format/.cd,%
				scaled y ticks = false,
				set thousands separator={\,},
				fixed},			
			enlarge y limits = {value = .2, upper},
			]
			
			\addplot[mark=*, OliveGreen, mark options={fill=OliveGreen}, nodes near coords] coordinates {(2, 1 376.89) (4, 4 365.74) (6, 7 446.32) (8, 10 652.79) (10, 14 499.20)};
			
			\end{axis}
		\end{tikzpicture}
	\end{minipage}
	\hfill
	\begin{minipage}[t]{0.49\linewidth}			
		\begin{tikzpicture}
			\begin{axis} [
			font=\footnotesize,
			nodes near coords style={above},
			every node near coord/.append style={
				align=center,
				font=\tiny,
				rotate=0,
				xshift=-5pt,
				color=black
			},
			height=4.5cm,
			width=\textwidth,
			xlabel={\textbf{Discord length range}, \(\times 128\)},
			ylabel style = {text width=4.5cm, align=center},
			ylabel={\textbf{Average runtime\\per discord, s}},
			enlarge y limits = {value = .2, upper},
			]
			
			\addplot[mark=*, OliveGreen, mark options={fill=OliveGreen}, nodes near coords] coordinates {(2, 10.76) (4, 11.37) (6, 11.63) (8, 11.89) (10, 12.59)};
	
			\end{axis}
		\end{tikzpicture} 
		\end{minipage}
	\caption{Synthetic dataset (RandomWalk1M)}
	\end{subfigure}
	\vspace{5pt}
	\caption{Scalability of the \ParMERLIN{} algorithm w.r.t. the discord length range}
	\label{fig:Scalability-DiscordLengthRange}
\end{figure}

%% file: parmerlin-casestudies.tex

\section{Case Study}
\label{sec:CaseStudies}

In this section, we apply \ParMERLIN{} to discover subsequence anomalies in a  real-world time series from a smart heating control system. The PolyTER system~\cite{ZymblerKLKShB20} allows for intelligent monitoring and control of operating conditions of utility systems through the analysis of the data from various IoT sensors installed in the university campus buildings. We took a time series from a temperature sensor installed in a lecture hall and discovered the anomalies in a specified range. The sensor's frequency is 4~times per hour, the time series corresponds to annual measurements (i.e., time series length $n=35~040$), and we search for anomalies that range from 12~hours to 7~days (i.e., $\textit{minL}=48$ and $\textit{maxL}=672$, respectively).

To visualize the results obtained, we propose the discord heatmap technique that illustrates the anomaly score through the intensity of a color, and is somewhat like the motif heatmap~\cite{DBLP:conf/icbk/MadridIMZSK19}. 
Formally speaking, 
we plot a one-color heatmap as a matrix of the size $(\textit{maxL}-\textit{minL}+1)$~$\times$ $(n-\textit{minL})$, where the intensity of a pixel $(m,i)$ shows the anomaly score of the discord $T_{i,\,m} \in D_m$, and the pixel's intensity is calculated as a 
normalization of the discord's distance to its nearest neighbor:
\begin{equation}
	\label{eq:DiscordHeatmap}
	\textit{heatmap}(m,\,i)=\frac{T_{i,\,m}.\textit{nnDist}}{2m},
\end{equation}
where we employ the normalizing divisor $2m$ according to Equation~\ref{eq:EDnorm}.

Despite the fact that we have proposed a visual tool to explore and discover multiple length discords, there is an open question for a practitioner, how we can rank discords of different lengths and extract the most interesting ones. There is a number of a discord attributes that can be suggested to take into account, e.g., its length and index, a distance to its nearest neighbor, a number of its self-matches, etc. However, in this study, we employ a straight-forward approach that considers a discord's interest as a normalized distance to its nearest neighbor comparing such distances among discords having the same index. Thus, the most interesting discord among the ones of different lengths is selected as below:
\begin{equation}
	\label{eq:DiscordInterest}
	\arg \max\limits_{1 \leq i \leq N} \max \limits_{\textit{minL} \leq m \leq \textit{maxL}} \textit{heatmap}(m,\,i).
\end{equation}
Clearly, through Equation~\ref{eq:DiscordInterest}, we can select top\nobreakdash-$k$ interesting discords. In our repository~\cite{KraevaZ22}, the reader can find plots of discord heatmaps and top discords for all the real-world time series listed in Tab.~\ref{tab:Datasets} and these ones from other subject domains.

\begin{figure}[ht!]
	\centering
	\captionsetup{justification=centering}
	\begin{subfigure}{\textwidth}	
		\includegraphics[width=\linewidth]{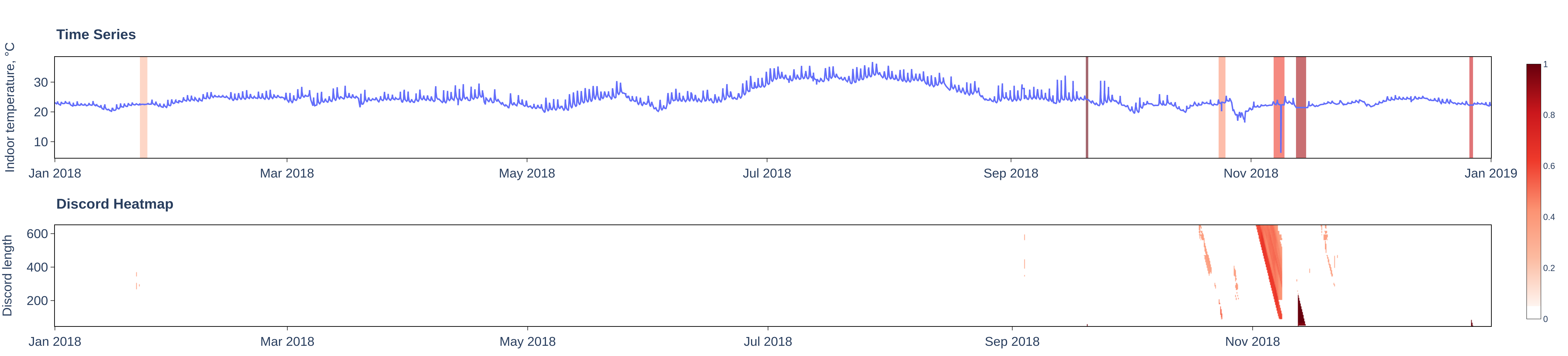} 
		\caption{Time series and its discord heatmap}
		\label{subfig:PolyTER-Heatmap} 
	\end{subfigure}
	\hfill
	\begin{subfigure}{\textwidth}
		\begin{minipage}[t]{0.25\linewidth}	
			\vspace{0pt}
			\includegraphics[width=\linewidth]{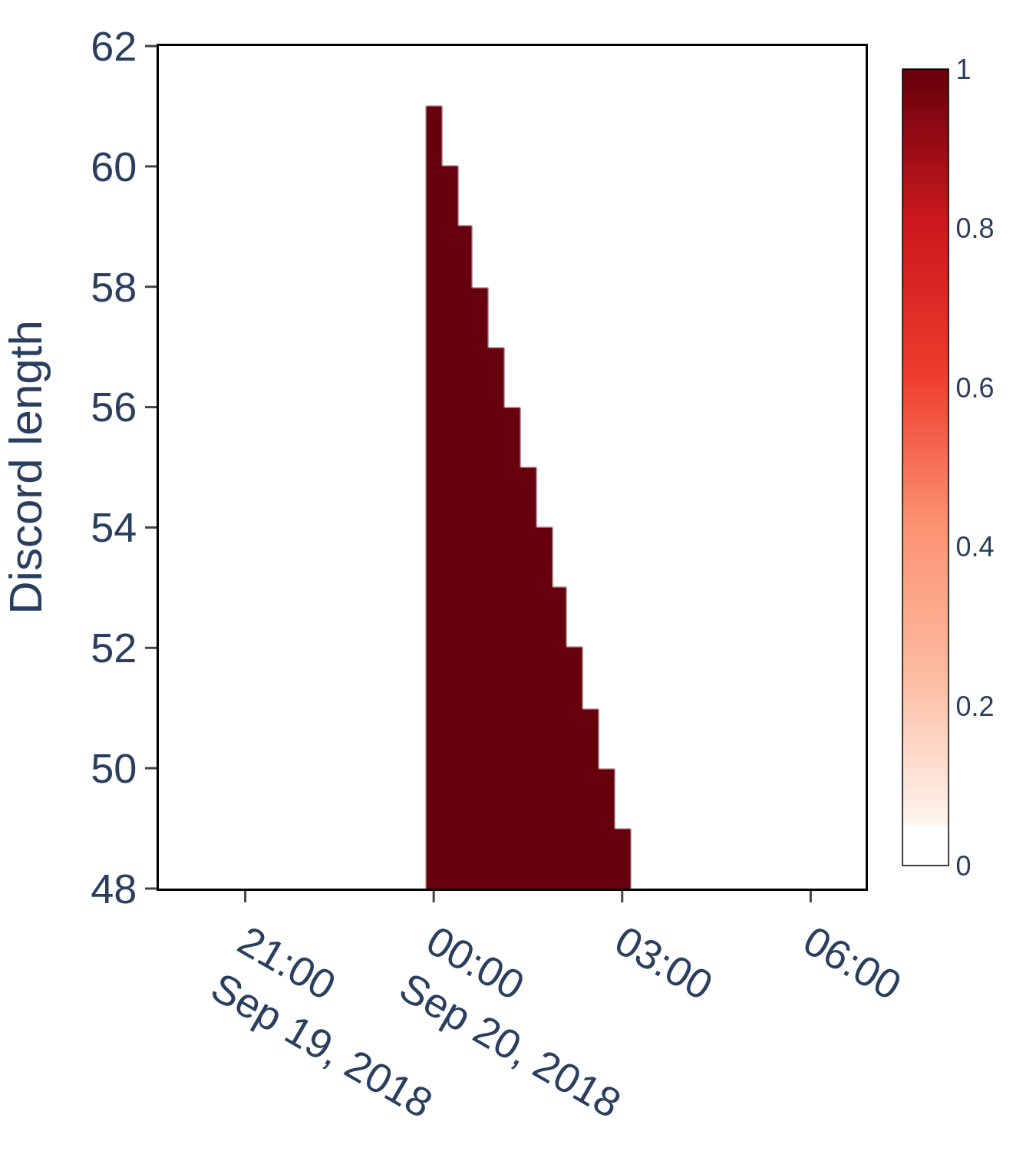}
		\end{minipage}
		\hfill
		\begin{minipage}[t]{0.5\linewidth}
			\vspace{0pt}
			\includegraphics[width=\linewidth]{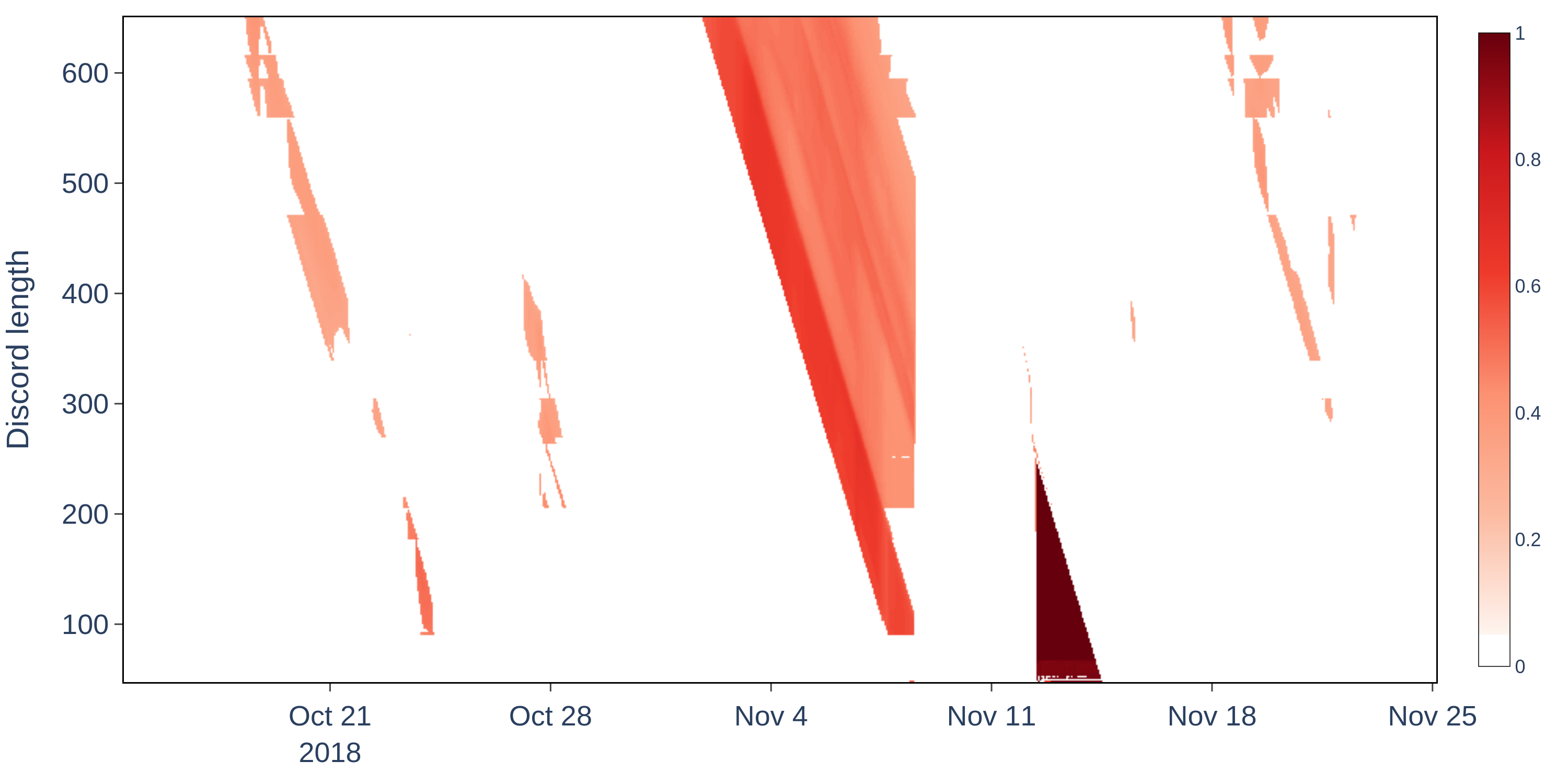}
		\end{minipage}
		\hfill
		\begin{minipage}[t]{0.22\linewidth}
			\vspace{0pt}
			\includegraphics[width=\linewidth]{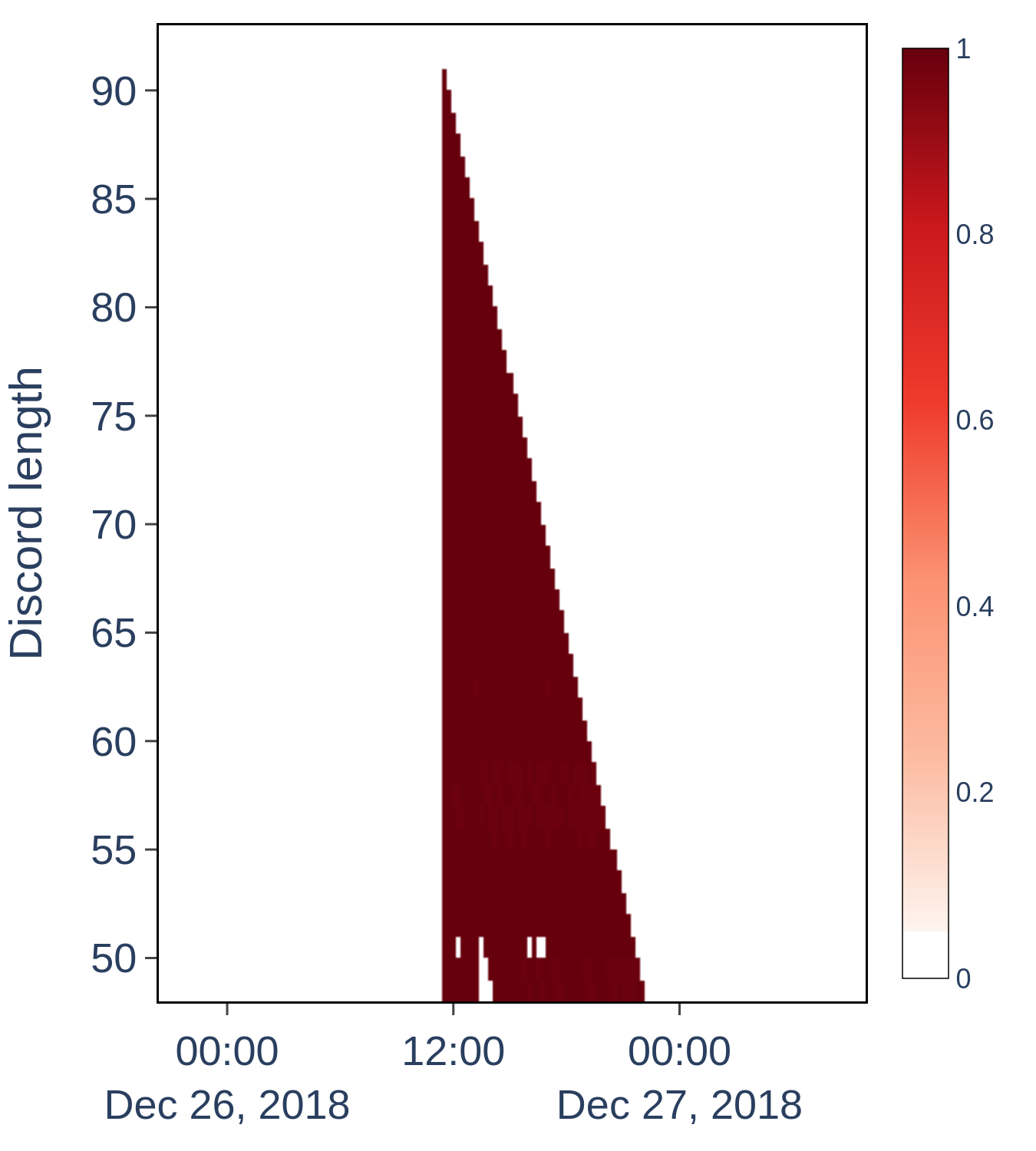}
		\end{minipage}
		\caption{Zooming the most interesting intervals of the discord heatmap}
		\label{subfig:PolyTER-Zoom} 
	\end{subfigure}	
	\hfill	
	\begin{subfigure}{\textwidth}	
		\begin{minipage}[h]{0.32\linewidth}	
			\includegraphics[width=\linewidth]{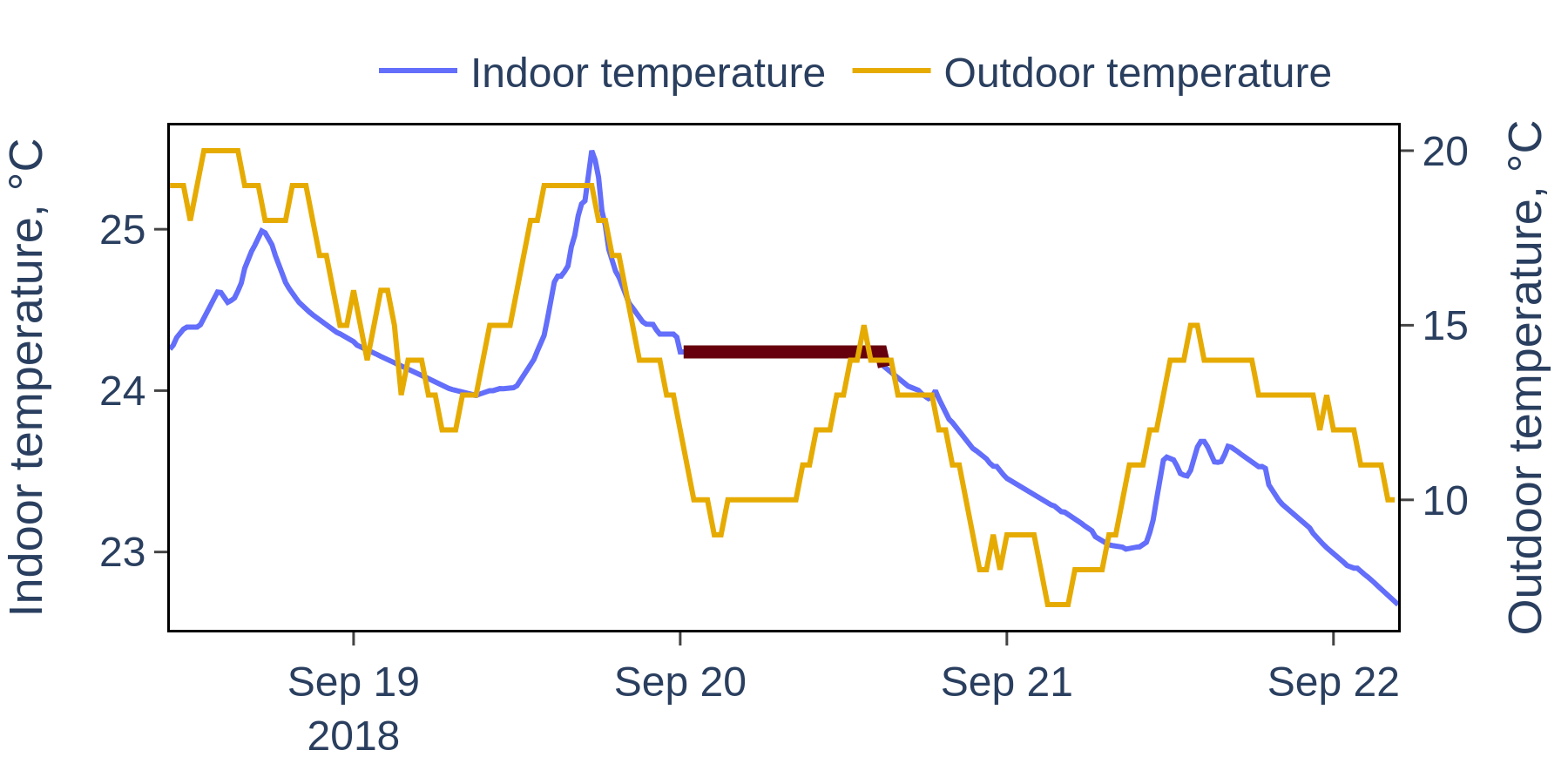}
			\caption*{Top-1 discord, $m=60$\\(15~hr)}
		\end{minipage}
		\hfill
		\begin{minipage}[h]{0.32\linewidth}
			\includegraphics[width=\linewidth]{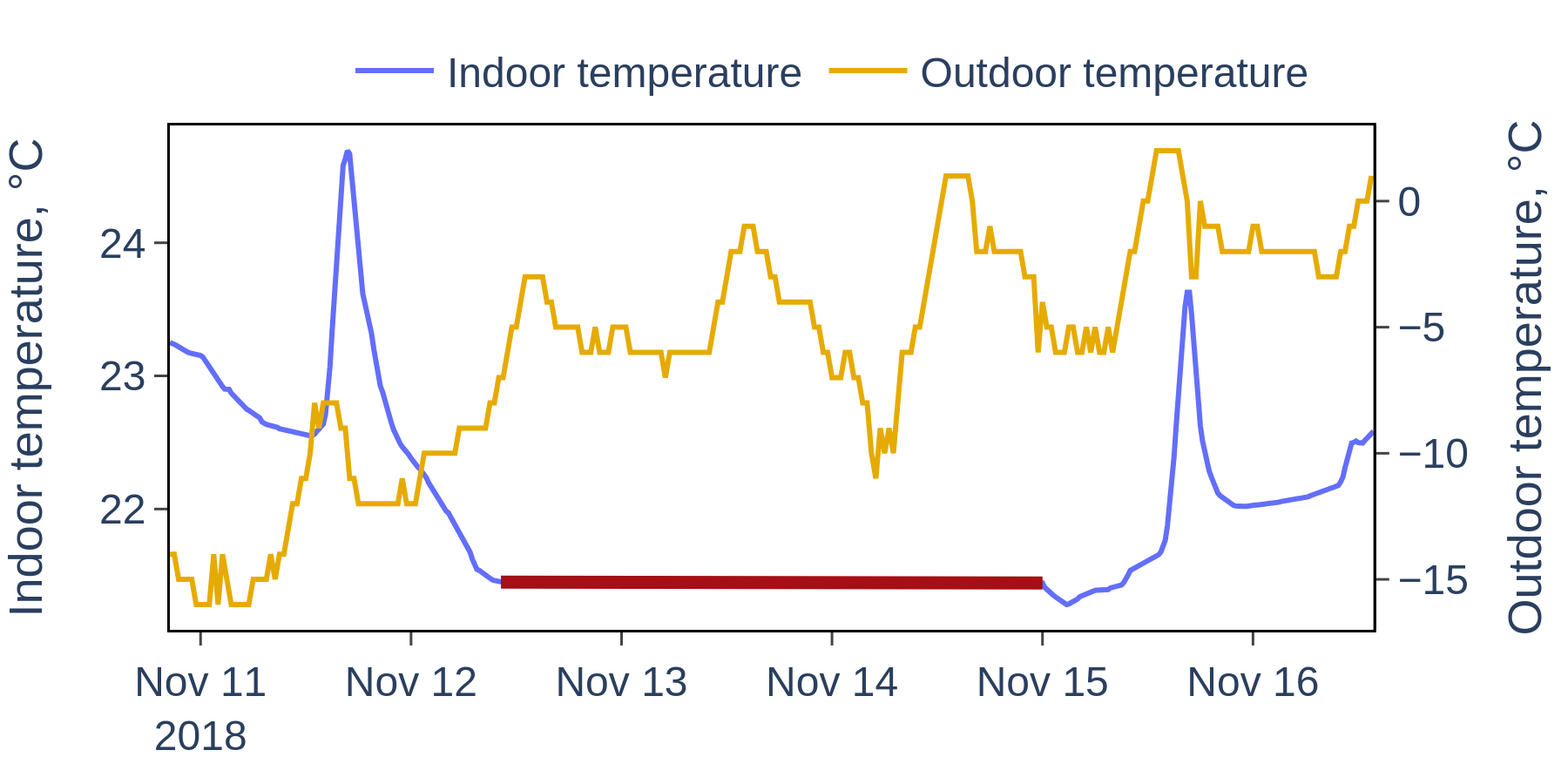}
			\caption*{Top-2 discord, $m=248$\\(2~days, 14~hr)}
		\end{minipage}
		\hfill
		\begin{minipage}[h]{0.32\linewidth}	
			\includegraphics[width=\linewidth]{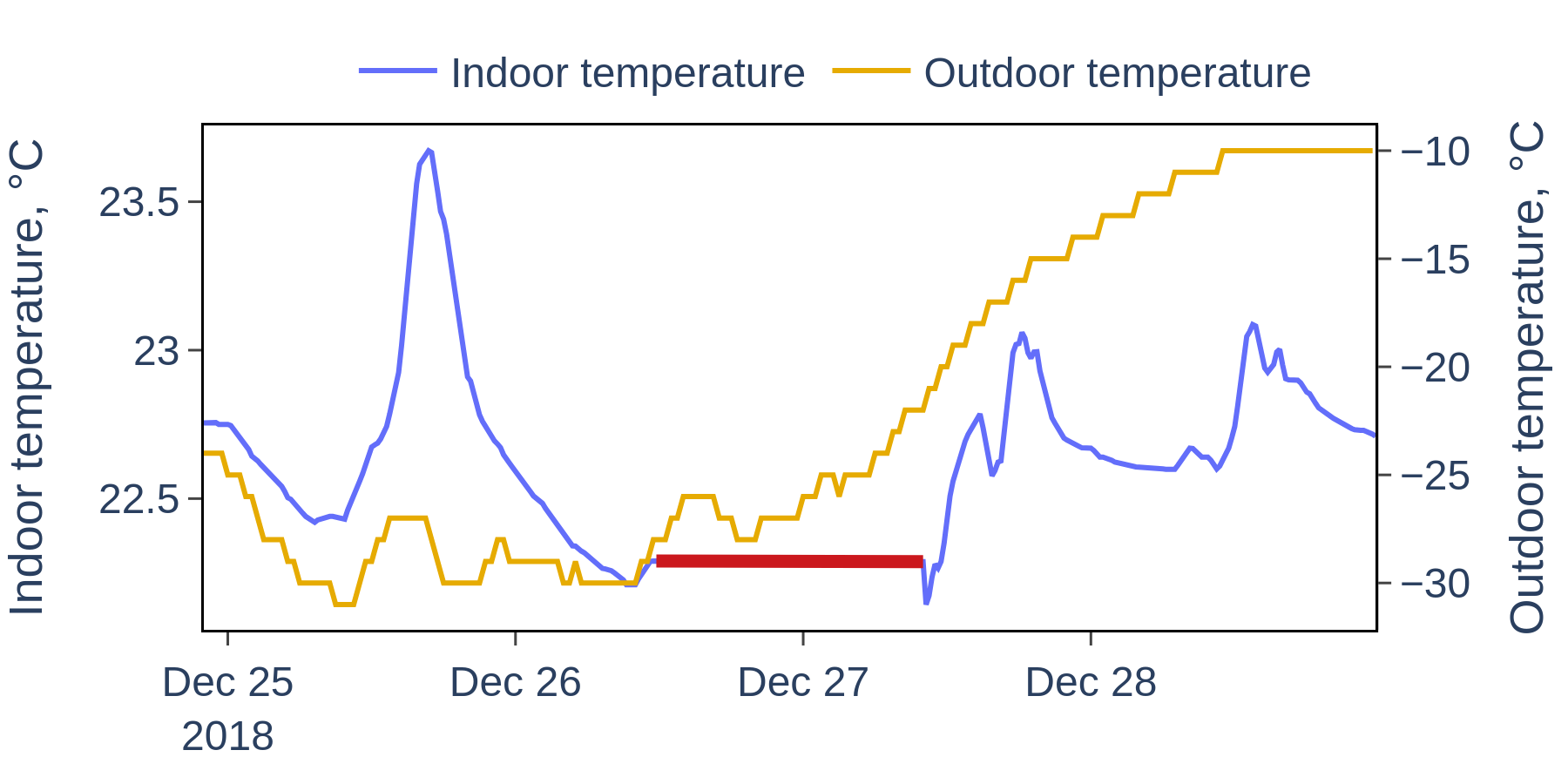}
			\caption*{Top-3 discord, $m=90$\\(22~hr 30~min)}
		\end{minipage}
		\vfill
		\begin{minipage}[h]{0.32\linewidth}	
			\includegraphics[width=\linewidth]{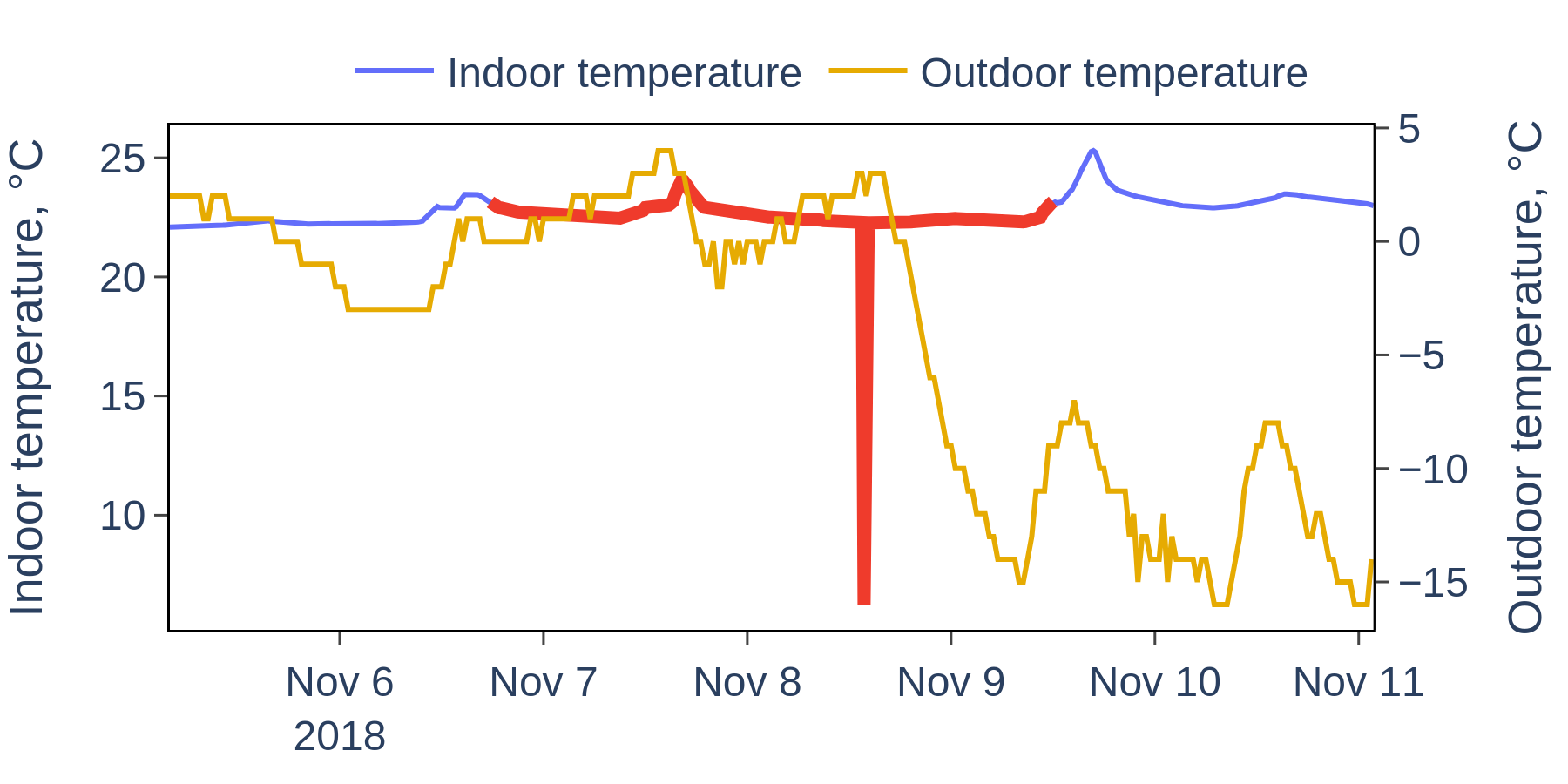}
			\caption*{Top-4 discord, $m=266$\\(2~days, 18~hr  30~min)}
		\end{minipage}
		\hfill
		\begin{minipage}[h]{0.32\linewidth}	
			\includegraphics[width=\linewidth]{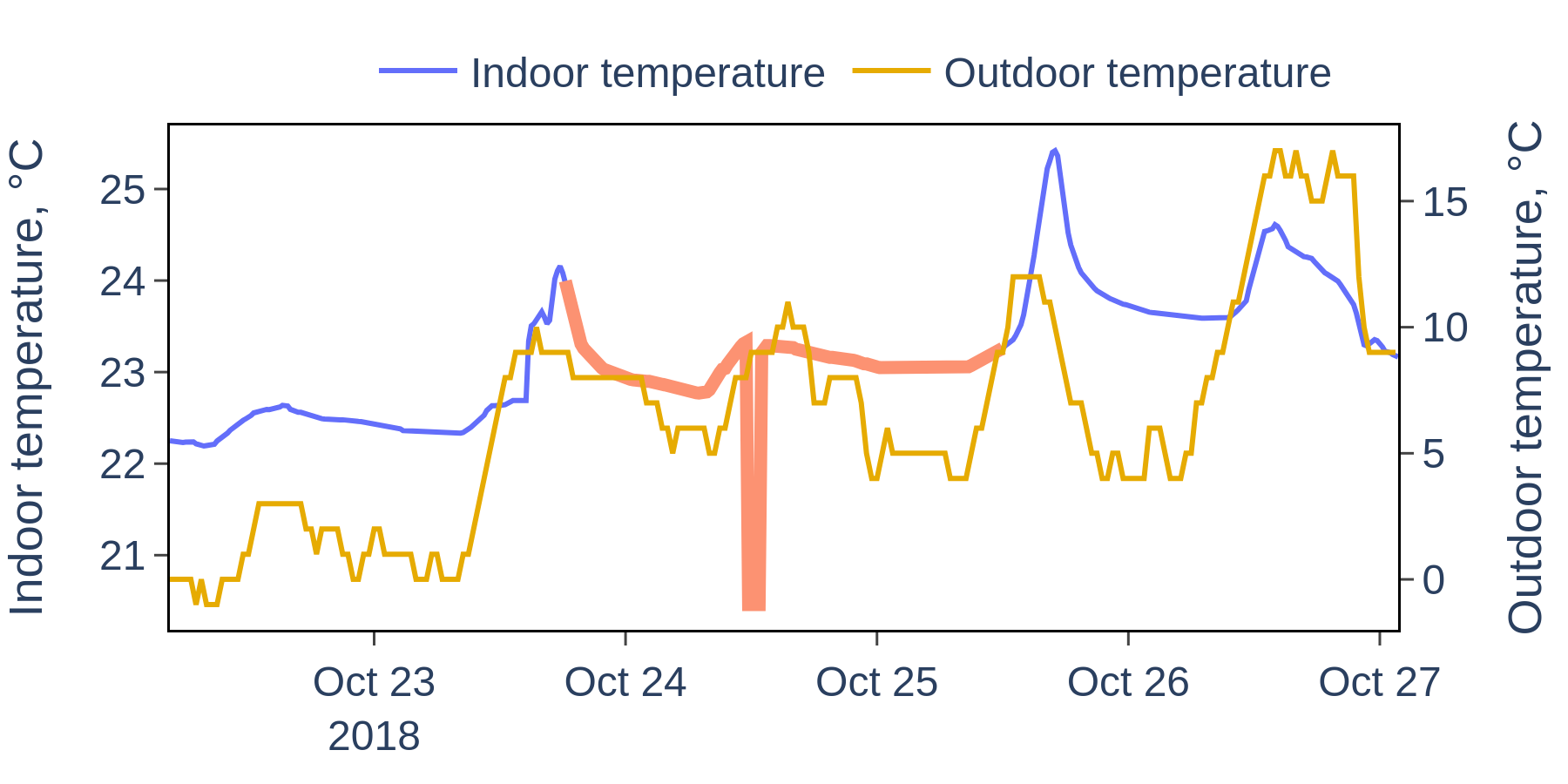}
			\caption*{Top-5 discord, $m=168$\\(1~day, 18~hr)}
		\end{minipage}
		\hfill
		\begin{minipage}[h]{0.32\linewidth}	
			\includegraphics[width=\linewidth]{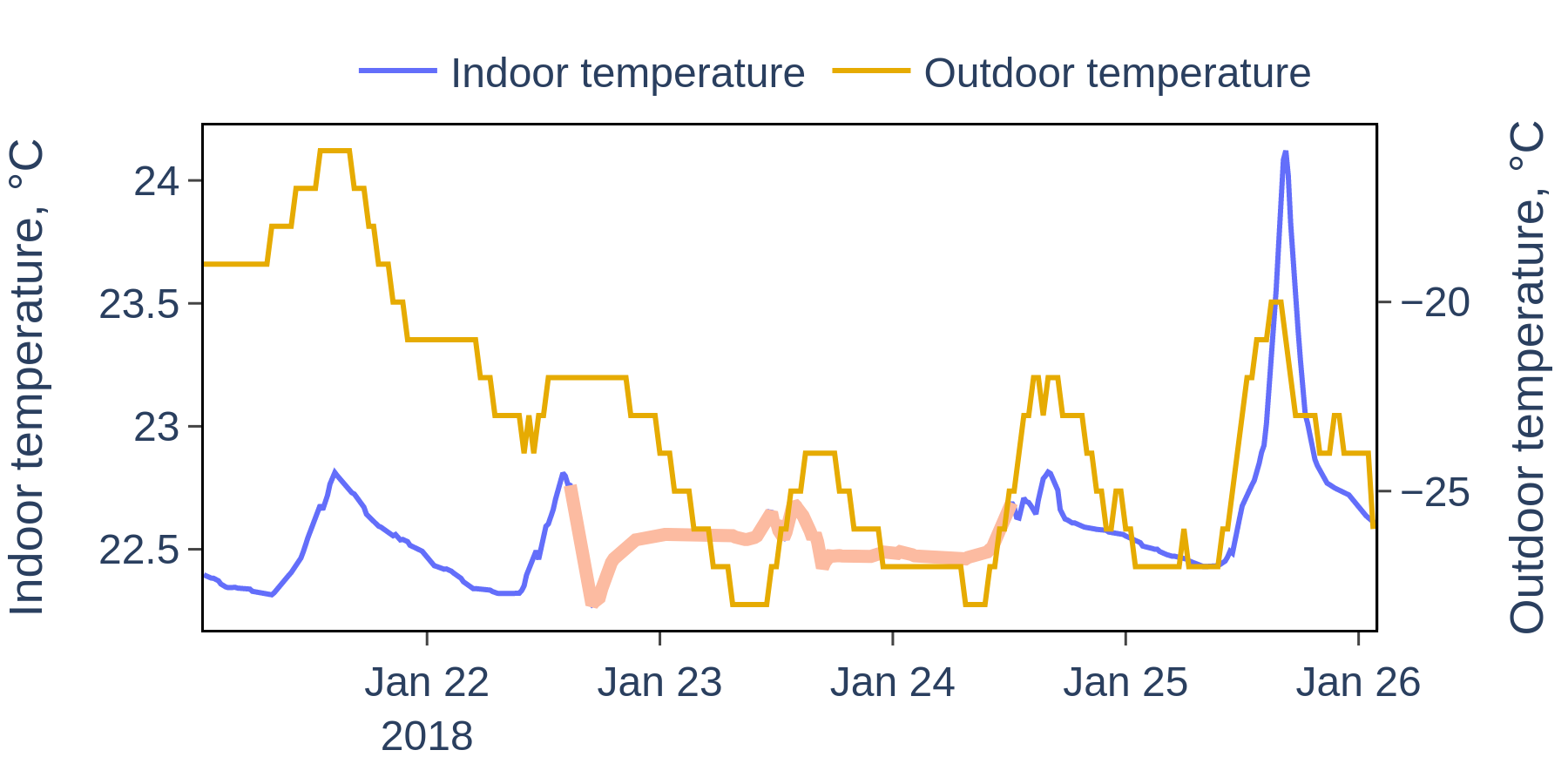}
			\caption*{Top-6 discord, $m=183$\\(1~day, 21~hr  45~min)}
		\end{minipage}
		\vfill
		\caption{Top-6 discords of different lengths}
		\label{subfig:PolyTER-TopDiscords} 
	\end{subfigure}
	\vspace{5pt}
	\caption{Case of PolyTER}
	\label{fig:PolyTER}
\end{figure}

Figure~\ref{fig:PolyTER} summarizes the results of the case study showing the time series and its discord heatmap, zooming of the heatmap intervals with most interesting discords, and top\nobreakdash-6 discords according to Equation~\ref{eq:DiscordInterest} (see sub-figures \ref{subfig:PolyTER-Heatmap}, \ref{subfig:PolyTER-Zoom}, and \ref{subfig:PolyTER-TopDiscords}, respectively). In addition, in top discord plots, we indicate both indoor and outdoor temperature, where the latter is obtained 
from the open weather archive~\cite{WeatherChelyabinsk22}. 
Highly likely, each of top three discords illustrates a long-term malfunction in a temperature sensor that outputs the same measurements during a considerable time period. Next, top\nobreakdash-4 and top\nobreakdash-5 discords show a short-term failure of the sensor. Finally, top\nobreakdash-6 discord may indicate the fact that the system (or its operator) chose an inefficient heating mode of the lecture hall in the considered period.

%% file: parmerlin-conclusion.tex

\section*{Conclusions}
\label{sec:Conclusion}

In this article, we addressed the problem of accelerating the time series subsequence anomaly discovery on a graphics processor. Such an anomaly refers to successive points in time whose collective behavior is abnormal, although each observation individually does not necessarily deviates. Currently, discovering subsequence anomalies in time series remains one of the most topical research problems. 

Among a large number of approaches to discovering subsequence anomalies, the discord concept~\cite{DBLP:conf/cbms/LinKFH05} is considered one of the best. A time series discord is intuitively defined as a subsequence that is maximally far away from its nearest neighbor. However, application of discords is reduced by sensitivity to a user's choice of the subsequence length. A brute-force discovering discords of all the possible lengths and then selecting the best discords with respect to some measure is clearly computationally prohibitive. Recently introduced the MERLIN algorithm~\cite{DBLP:conf/icdm/NakamuraIMK20} discovers time series discords of every possible length in a specified range, being ahead of competitors in terms of accuracy and performance. MERLIN employs repeated calls of the DRAG algorithm~\cite{DBLP:conf/icdm/YankovKR07} that discovers discords of a given length with a distance of at least $r$ to their nearest neighbors, and adaptive selection of the parameter $r$. However, to the best of our knowledge, no research has addressed the accelerating MERLIN with any parallel hardware architecture.

In the article, based on Keogh \textit{et~al.}'s works~\cite{DBLP:conf/icdm/YankovKR07,DBLP:conf/icdm/NakamuraIMK20}, we proposed a novel parallelization scheme \ParMERLIN{} (\ParMERLINtranscript) for a graphics processor. When basically following the original serial algorithm, \ParMERLIN{}, however, employs our derived recurrent formulas to calculate the mean values and standard deviations of subsequences of the time series. Since that data are further involved in calculations of the normalized Euclidean distances between subsequences, eventually, we significantly reduce the amount of calculations. Furthermore, \ParMERLIN{} repeatedly calls \ParDRAG{} (\ParDRAGtranscript{})~\cite{KraevaZ23}, our developed parallel version of the original DRAG algorithm. Similar to its predecessor, \ParDRAG{} performs in two phases. To implement the candidate selection phase, we exploit the data parallelism by division the time series into equal-length segments, where each segment is processed separately by a  block of GPU threads. The thread block considers the segment subsequences as local candidates to discords and processes the subsequences that are located to the right of the segment and do not overlap with the candidates. Next, the thread block scans the subsequences in chunks, the number of elements in which is equal to the segment length, and the first chunk begins with the $m$\nobreakdash-th element in the segment, where $m$ is the discord length. Such a technique allows us for avoiding redundant checks of the fact that candidates and subsequences in chunks overlap. In \ParDRAG{}, the candidate refinement phase is parallelized similar to the selection phase. Refinement involves only those segments of the time series whose set of local candidates is not empty. The algorithm scans and processes the subsequences that do not overlap with the candidates and are located to the left of the segment.

We carried out an extensive experimental evaluation of \ParMERLIN{} over real-world and synthetic time series. In the experiments, we compared our development with two  algorithms, namely, KBF\_GPU~\cite{ThuyAC21} and Zhu \textit{et~al.}'s~\cite{DBLP:journals/tpds/ZhuJGD21} since our thorough review of related work did not reveal other GPU-oriented parallel competitors. Both rivals aim at discovery top\nobreakdash-1 discord, where the former is a parallelization of the brute-force approach while the latter employs computational patterns to reduce amount of calculations. As expected, in the experiments, our algorithm significantly outruns KBF\_GPU. Next, being adapted to discover the most important discord of a specified length, Zhu \textit{et~al.}'s algorithm significantly outruns \ParMERLIN{} that discovers discords of every possible length in a specified range: up to 20~times and up to two orders of magnitude greater over real and synthetic time series, respectively. However, \ParMERLIN{} discovers substantially more discords: at least to two and seven orders of magnitude greater over real and synthetic time series, respectively. Thus, \ParMERLIN{} significantly outruns the rivals in terms of the average running time to discover one discord. Finally, in the experiments, we also investigated the scalability of \ParMERLIN{} and found that the algorithm's running time is proportional to each of the following parameters for both real-world and synthetic time series: the segment length, time series length, and discord range length. 

We also apply \ParMERLIN{} to discover anomalies in a real-world time series from a smart heating control system employing our proposed discord heatmap technique to illustrate the results.

Our further studies might elaborate on the following topics: (a) discords discovery in a large time series that cannot be entirely placed in RAM with a high-performance cluster based on GPU nodes, and (b) application of \ParMERLIN{} in a deep learning-based online time series anomaly detection. 

%% file: parmerlin-appendix.tex

\section*{Appendix}
\label{sec:Appendix}



\begin{lemma}
Let us have the time series $T$, $|T|=n$, and two its $m$\nobreakdash-length subsequences, $T_{i,\,m}$ and $T_{i,\,m+1}$, where $1 \leq i \leq n-m$ and $3 \leq m \ll n$. Then the following holds:
\[	\mu_{T_{i,\,m+1}}=\dfrac{1}{m+1}\bigl(m\mu_{T_{i,\,m}} +t_{i+m}\bigr), \]
\[	\sigma^2_{T_{i,\,m+1}}=\dfrac{m}{m+1} \Bigl( \sigma_{T_{i,\,m}}^2 + \dfrac{1}{m+1} \bigl(\mu_{T_{i,\,m}}-t_{i+m}\bigr)^2\Bigr). \]	
\end{lemma}

\begin{proof}
First, let us prove the equation on the mean value. According to the definition
\[	\mu_{T_{i,\,m}}=\dfrac{1}{m}\sum_{k=0}^{m}t_{i+k}.\]
Then
\[	\sum_{k=0}^{m}t_{i+k}=m\mu_{T_{i,\,m}}.\]
Next, let us consider $\mu_{T_{i,\,m+1}}$:
\[ \mu_{T_{i,\,m+1}}=\dfrac{1}{m+1}\sum_{k=0}^{m}t_{i+k}=\dfrac{1}{m+1}(m\mu_{T_{i,\,m}} + t_{i+m}), 
\] 
so, the first equation is proved. 
Further, let us prove the equation on the standard deviation. According to the definition
\[\sigma^2_{T_{i,\,m}}=\dfrac{1}{m} \sum_{k=0}^{m-1}t^2_{i+k}-\mu^2_{T_{i,\,m}}. 
\]
Then 
\[
\sum_{k=0}^{m-1}t^2_{i+k}=m(\sigma^2_{T_{i,\,m}}+\mu^2_{T_{i,\,m}}). 
\]
Next, let us consider $\sigma^2_{T_{i,\,m+1}}$:
\[ \sigma^2_{T_{i,\,m+1}}=\dfrac{1}{m+1}\sum_{k=0}^{m}t^2_{i+k}-\mu^2_{T_{i,\,m+1}}=\dfrac{1}{m+1}\bigl(m(\sigma^2_{T_{i,\,m}}+\mu^2_{T_{i,\,m}})+t^2_{i+m}\bigr)-\mu^2_{T_{i,\,m}}.
\]
Employing the above-proved proposition on the mean value, we obtain
\[ \sigma^2_{T_{i,\,m+1}}=\dfrac{1}{m+1}\bigl(m(\sigma^2_{T_{i,\,m}}+\mu^2_{T_{i,\,m}})+t^2_{i+m}\bigr)-\bigl(\dfrac{1}{m+1}(m\mu_{T_{i,\,m}} + t_{i+m})\bigr)^2.
\]
Performing operations in parentheses and further collecting terms, we obtain
\[
\sigma^2_{T_{i,\,m+1}}=\dfrac{1}{m+1}\bigl(m(\sigma^2_{T_{i,\,m}}+\mu^2_{T_{i,\,m}})+t^2_{i+m}-\dfrac{1}{m+1}(m\mu_{T_{i,\,m}} + t_{i+m})^2\bigr)=
\]
\[
=\dfrac{1}{m+1}\Bigl(m\sigma^2_{T_{i,\,m}} + \dfrac{1}{m+1}\bigl(m(m+1)\mu^2_{T_{i,\,m}} - m^2\mu_{T_{i,\,m}} + (m+1)t^2_{i+m}-t^2_{i+m} - 2m\mu_{T_{i,\,m}}t_{i+m} \bigr)\Bigr)=
\]
\[
=\dfrac{1}{m+1}\Bigl(m\sigma^2_{T_{i,\,m}} + \dfrac{m}{m+1}\bigl(\mu^2_{T_{i,\,m}} - 2\mu_{T_{i,\,m}}t_{i+m} + t^2_{i+m}  \bigr)\Bigr)=
\]
\[
=\dfrac{m}{m+1}\Bigl(\sigma^2_{T_{i,\,m}} + \dfrac{1}{m+1}\bigl(\mu_{T_{i,\,m}} - t_{i+m}  \bigr)^2\Bigr).
\]
This concludes our proof.
\end{proof}